\documentclass[amsmath,nobibnotes,aps,superscriptadress,amssymb,pra,aps,showpacs,superscriptaddress, twocolumn]{revtex4-1}

\usepackage{hyperref}
\hypersetup{colorlinks,citecolor=blue,urlcolor=blue,bookmarks=false,hypertexnames=true}
\usepackage[usenames]{color}
\usepackage{graphicx}
\usepackage{amsthm}
\usepackage{caption}
\usepackage{subfigure}
\usepackage{amsmath}
\usepackage{epsfig}

\usepackage{dcolumn}
\usepackage{lmodern}
\usepackage{enumitem}
\usepackage{bm}
\usepackage{color}
\usepackage{times}
\usepackage{epstopdf}
\usepackage{amssymb}
\usepackage{amstext}
\usepackage{latexsym}
\usepackage{hyperref}
\usepackage{amsfonts}
\usepackage{psfrag}
\usepackage{soul,xcolor}
\usepackage[normalem]{ulem}
\usepackage{dsfont}
\usepackage[T1]{fontenc}
\usepackage{tgbonum}
\usepackage{txfonts}
\usepackage{footnote}
\usepackage{multirow}
\usepackage{appendix}
\usepackage{mathtools}
\usepackage{xspace}
\usepackage{braket}
\usepackage{algorithm2e}
\usepackage{color}
\newtheorem{theorem}{Theorem}

\newtheorem{corollary}{Corollary}[theorem]

\usepackage{mathtools}

\begin{document}
\setstcolor{red}
\newtheorem{Proposition}{Proposition}[section]	

\title{On Simultaneous Information and Energy Transmission through Quantum Channels} 
\date{\today}

\author{Bishal Kumar Das}
\email{bishal.9197@physics.iitm.ac.in}
\affiliation{Department of Physics, Indian Institute of Technology Madras, Chennai, India, 600036}
\affiliation{Centre for Quantum Information, Communication, and Computing (CQuICC), Indian Institute of Technology Madras, Chennai, India 600036}

\author{Lav R.\  Varshney}
\email{varshney@illinois.edu}
\affiliation{Department of Electrical and Computer Engineering, University of Illinois Urbana-Champaign, Urbana, Illinois 61801, USA}

\author{Vaibhav Madhok}
\email{madhok@physics.iitm.ac.in}
\affiliation{Department of Physics, Indian Institute of Technology Madras, Chennai, India, 600036}
\affiliation{Centre for Quantum Information, Communication, and Computing (CQuICC), Indian Institute of Technology Madras, Chennai, India 600036}

\begin{abstract}
The optimal rate at which information can be sent through a quantum channel when the transmitted signal must simultaneously carry some minimum amount of energy is characterized.  To do so, we introduce the quantum-classical analogue of the capacity-power function and generalize results in classical information theory for transmitting classical information through noisy channels. We show that the capacity-power function for a classical-quantum channel, for both unassisted and private protocol, is concave and also prove additivity for unentangled and uncorrelated ensembles of input signals for such channels. This implies we do not need regularized formulas for calculation. We show these properties also hold for all noiseless channels when we restrict the set of input states to be pure quantum states. For general channels, we geometrically prove that the capacity-power function is piecewise concave, with further support from numerical simulations.
 Further, we connect channel capacity to properties of random quantum states. In particular, we obtain analytical expressions for the capacity-power function for the case of noiseless channels using properties of random quantum states under an energy constraint and concentration phenomena in large Hilbert spaces. 
\end{abstract}
\maketitle

\section{Introduction}

How physical systems process information and exchange energy is crucial to gaining insights into the workings of our universe. For example, the connections between entropy, information, and thermodynamics form the cornerstone of statistical mechanics. 
The study of phase transitions that involve an interplay between free energy and entropy, or Maxwell's demon paradox and its resolution via Landauer's principle are important examples. 
``Information is physical" as stated by Landauer \cite{launder1991} implies that the information processing capability of a device is inseparable from its physical properties. Quantum information processing has  extended these ideas to the quantum domain \cite{nielsenchuang,wilde_book}. 
However, there is very little work considering quantum communication under physical constraints like energy. 

The importance of simultaneous transmission of energy and information, two fundamental physical quantities, or the transmission of information under any constraint is of central importance in quantum information theory.

For example, quantum metrology deals with maximizing Fisher information 
in the presence of losses and under constraints of finite laser power  and mean photon number. For example, let pure states be used in the arms of the interferometer and $\varphi$ be the parameter to be estimated. 
Then $\varphi$ dependent changes in the state in the interferometer is generated by the Hamiltonian, 
\begin{equation}
	\label{eq:Hprobe}
	H_{\rm probe} = \varphi H.
\end{equation}
 The quantum Cram\'{e}r-Rao bound bound tells us that the measurement uncertainty is bounded from below as
\begin{equation}
	\label{eq:crb1}
	\delta \varphi \geq \frac{1}{2 \sqrt{\langle \Delta^{2} H \rangle}},
\end{equation}
where $\langle \Delta^{2}H \rangle$ is the variance of $H$ \cite{caves_braunstein}. 
It turns out that the mean photon number of the state in a wave-based interferometer with additional constraint like the requirement that the states in the two modes in the interferometer be both Gaussian constrains the minimum achievable uncertainty and therefore maximizes the Fisher information.


Rather than considering quantum metrology under constraints, here we consider quantum communication under power requirements.
 There are quantities like the mutual information, as fundamental as the Fisher information, which capture the rate of communication between two parties. 
In this work, we are interested in maximizing channel capacity under constraints, which is intimately related to mutual information. 

Classical information theory has shown a fundamental tradeoff between transmitting information and energy simultaneously  {using the same patterned energy signal}, defined over a set of symbols that have different energy values \cite{varshney2008transporting, grover2010shannon}.  To transmit maximal information, the capacity-achieving input distribution over symbols should be used, whereas to maximize energy transmission, the most energetic symbol should be used all the time.  There is a tension in trying to do both with the same signal.  Although such information-theoretic analysis has inspired numerous practical engineering (communication and energy) systems such as energy-harvesting communications or powerline communications in the energy grid  \cite{clerckx2021wireless}, it also addresses the fundamental relationship among basic physical quantities such as energy, matter, and information \cite{varshney2012}.

\begin{figure*}
\centering

\includegraphics[height = 9cm, width = 17cm]{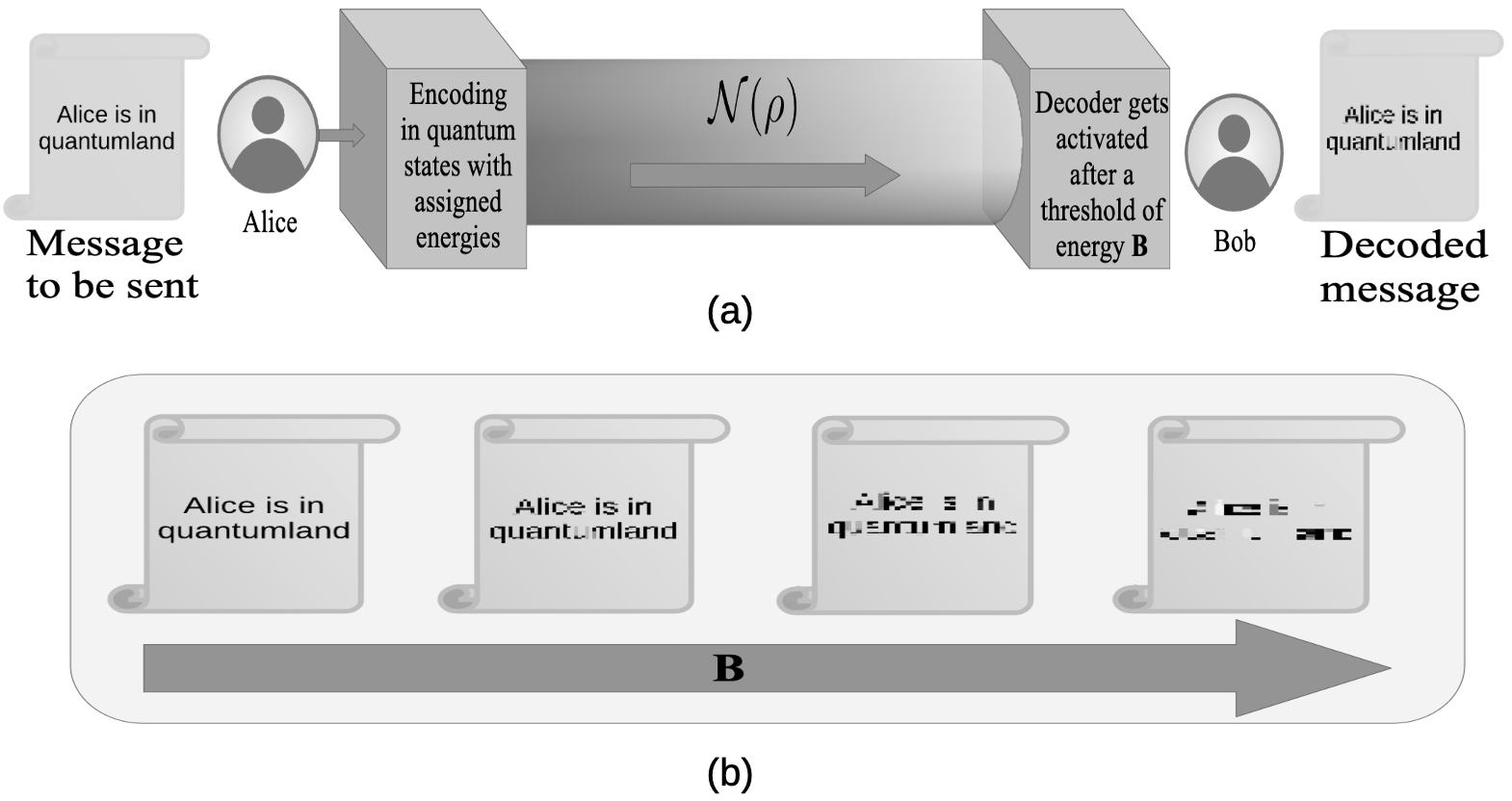}
\caption{Graphical summary: The effect of threshold energy on the  information at the receiver: (a) The information is encoded in quantum states, which are sent through a physical channel. At the output, a detector captures the states and decodes the information. However, the detector is activated only after a certain threshold of energy $B$, carried by the output states. (b) A pictorial representation of the effect of an increase in $B$ on the decoded information. As the threshold of energy $B$ increases, information content at the receiver becomes poorer as shown by the increasingly mangled message at the receiver.}
\label{graphical_represntation}
\end{figure*}

Communication has been a central topic in quantum information science: Sending classical or quantum information through a quantum channel has been studied extensively \cite{Holevo2002,Holevo_some_estimates,holevo2004entanglement,Caves_bosonic_RevModPhys.66.481,gyongyosi2018survey,Lyod_PhysRevA.55.1613,Shor_2003,Devetak2005,Additivity_of_the_classical_capacity_Shor, Schaumacher,HSW, holevo_gen_signal,Evaluating_capacities_of_bosonic_Gaussian_channels,wilde_book,nielsenchuang,preskilllecture,hastings2009superadditivity}. Quantum communication is often performed using photons to transmit information over quantum channels.  But photons also carry energy proportional to their wavelength. Here, we ask whether there is an analogous tension between transmitting information and energy simultaneously over quantum channels, as there is over classical channels. 
In particular, we characterize systems
that simultaneously meet two goals:
\begin{enumerate}
	\item Large received energy per unit time (large received power), and
	\item Large information per unit time.
\end{enumerate}

A classical picture of energy-constrained classical communication rate per unit cost was developed by Verdu \cite{verdu}, where he introduced the notion of \emph{capacity per unit cost function}, which is the capacity of a channel under a cost constraint on the input signal. He gave a simplified formula for the capacity of such a channel per unit cost when there is a free symbol. Extending the same formulations to classical-quantum channels \footnote{We will revisit these definitions later in the manuscript.}, some previous works consider sending classical information through quantum channels under a maximum energy constraint on the input.  Notably \cite{wildecapacitycost,PhysRevA.96.032340}  gave a simplified formula for capacity per unit cost when there exists a zero-cost state. 

Here we consider requiring minimum received power rather than constraining transmit power, and consider the entire tradeoff curve to define quantum analogues of previously studied classical capacity-power functions \cite{varshney2008transporting}. We give new results and proofs regarding their properties.
    In particular, we analytically prove that the capacity-power function is concave for single-use unassisted communication (i.e.\ without any pre-shared entanglement) for the classical-quantum channel and for private communication through any degradable classical-quantum channel. We further prove additivity for these channels, so we do not need regularized versions to calculate the capacity-power function.
 We show these properties also hold for all noiseless channels when we restrict the set of input states to be pure quantum states. For general channels, we find that the capacity-power function is piecewise concave. We give a visual proof for this supported by numerical simulations. Lastly, we show how entropy of random quantum states under an energy constraint can be related to noiseless capacity-power functions.

 In studying the information rate when imposing a condition that the channel output must carry at least some energy,
we calculate energy using  Hamiltonians that are  relevant in physical settings. We also give numerical examples to demonstrate our analytical calculations for three types of quantum channels: beam-splitting channel, depolarizing channel, and amplitude damping channel. Further, we obtain analytical expressions for the capacity-power functions for noiseless channels using properties of random quantum states.

This provides both engineering utility in understanding the limits of quantum communications and broader insight into the basic physical interplay between energy and information.  Figure \ref{graphical_represntation} is a schematic diagram of the protocol we consider. This work deals with the fundamental tradeoff between transmitting energy and transmitting information over a single noisy line.
The idea of energy-constrained private and quantum capacity was introduced in \cite{energy_constrained_private_and_quantum_capacity}. Unlike these past works and  traditional transmitter power constraints in information theory, where small transmitted power is desired, here large received power is desired. Except for works on reversible computing \cite{Varshney2010}, the fact that matter/energy must go along with information does not seem to have been considered in information theory. Similarly, the information carried in power transmission seems not to have been considered in power engineering.

\section{Background}

\begin{figure}
    \centering
    \includegraphics[scale = 0.25]{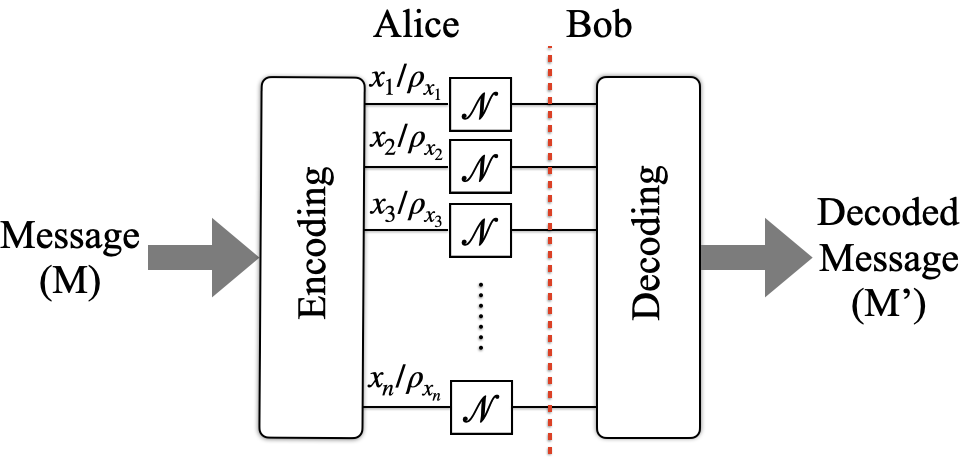}
    
    \caption{Typical protocol  for communication through a quantum channel with blocklength $n$.}
    \label{fig:pictorial representation of channel}
\end{figure}
To explain the information-energy transmission setting, let us review the classical protocol that maximizes channel capacity under the constraint that the expected energy of the output should be greater than a given threshold \cite{varshney2008transporting}.  First, we introduce some notation.  

Let $X$ be a random variable with mass function $p(\cdot)$ and $X^{n}$ be a vector random variable. 
 Let $x$ be the realized value taken by the random variable $X$ from the alphabet $\mathcal{X}$, and $x^{n}$ be the vector of values taken by $X^{n}$ from alphabet $\mathcal{X}^n$. The entropy of a random variable is $\mathcal{H}(X) =- \sum_{x\in \mathcal{X}}p(x)\ln p(x)$. For a binary random variable with parameter $p$, the binary entropy is $h(p) = -p\log p - (1-p)\log(1-p)$. For a given density matrix, $\rho$, the {von Neumann entropy is $S(\rho)=-tr (\rho \ln \rho)$.} The mutual information, classical as well as quantum, is denoted by $I(X:Y)$ and will be clear from context.

A diagram of communication through a classical channel is shown in Fig.~\ref{fig:pictorial representation of channel} with input alphabet  $\mathcal{X}$ and output alphabet $\mathcal{Y}$. For a block of length $n$, the channel input is $(X_1, X_2, \ldots, X_n) = X^{n}$ and the channel output is $(Y_1, Y_2, \ldots, Y_n) = Y^{n}$. The probability distribution of the input strings is $p_{X^{n}}(x_1, x_2, \ldots, x_n)$  where $(x_1, x_2, \ldots, x_n)$ is an instance of $X^{n}$. Similarly, $p_{Y^{n}}(y_1, y_2, \ldots, y_n)$ is the output probability distribution. Noise in the channel is characterized by the transition probability assignment $Q_{Y|X}(\cdot|\cdot)$.  Each output letter $y\in\mathcal{Y}$ has an associated energy function $b(y)$, and an energy constraint is applied to the output string. The average energy received is $E[b(Y^{n})] = \sum_{y^{n} \in \mathcal{Y}^{n}} p(y^{n}) b(y^{n})$. 

The energy constraint $E[b(Y^{n})] \geq n B$ reduces the set of admissible $p_{Y^{n}}$ and hence the set of admissible $p_{X^{n}}$.  An input vector $X^{n}$ that satisfies the constraint $E[b(Y^{n})] \geq nB$ is called $B$-admissible. Maximization is over all $B$-admissible input probability distributions.  Hence, the $n$th capacity-power function for a block of length $n$ is maximizing $I(X^{n}: Y^{n})$ over a reduced set of input and output probability distributions and is defined as,
\begin{equation}
C_n(B) = \underset{x^{n} \in \mathcal{X}^{n}\textcolor{blue}{\text{:}~}  E[b(Y^{n})] \geq n B}{\max} I(X^{n}: Y^{n}) \mbox{.}
\end{equation}
Intuitively, adding an energy constraint reduces the channel capacity, given by the optimized mutual information $I(X^{n}: Y^{n})$ of the unconstrained problem. 
The capacity-power function for the channel is defined as 
$$C(B) = \underset{n\rightarrow \infty}{\lim}\frac{1}{n} C_n(B)\mbox{.}$$

Before proceeding, let us discuss some examples of the classical capacity-power function and methods of calculation \cite{varshney2008transporting}.  If we have a classical memoryless binary noiseless channel without any energy constraint, then the capacity is $1$ and that capacity is achieved by the input distribution $P(0) = 1/2$ and $P(1) = 1/2$.

\begin{figure}
    \centering
    \includegraphics[scale = 0.13]{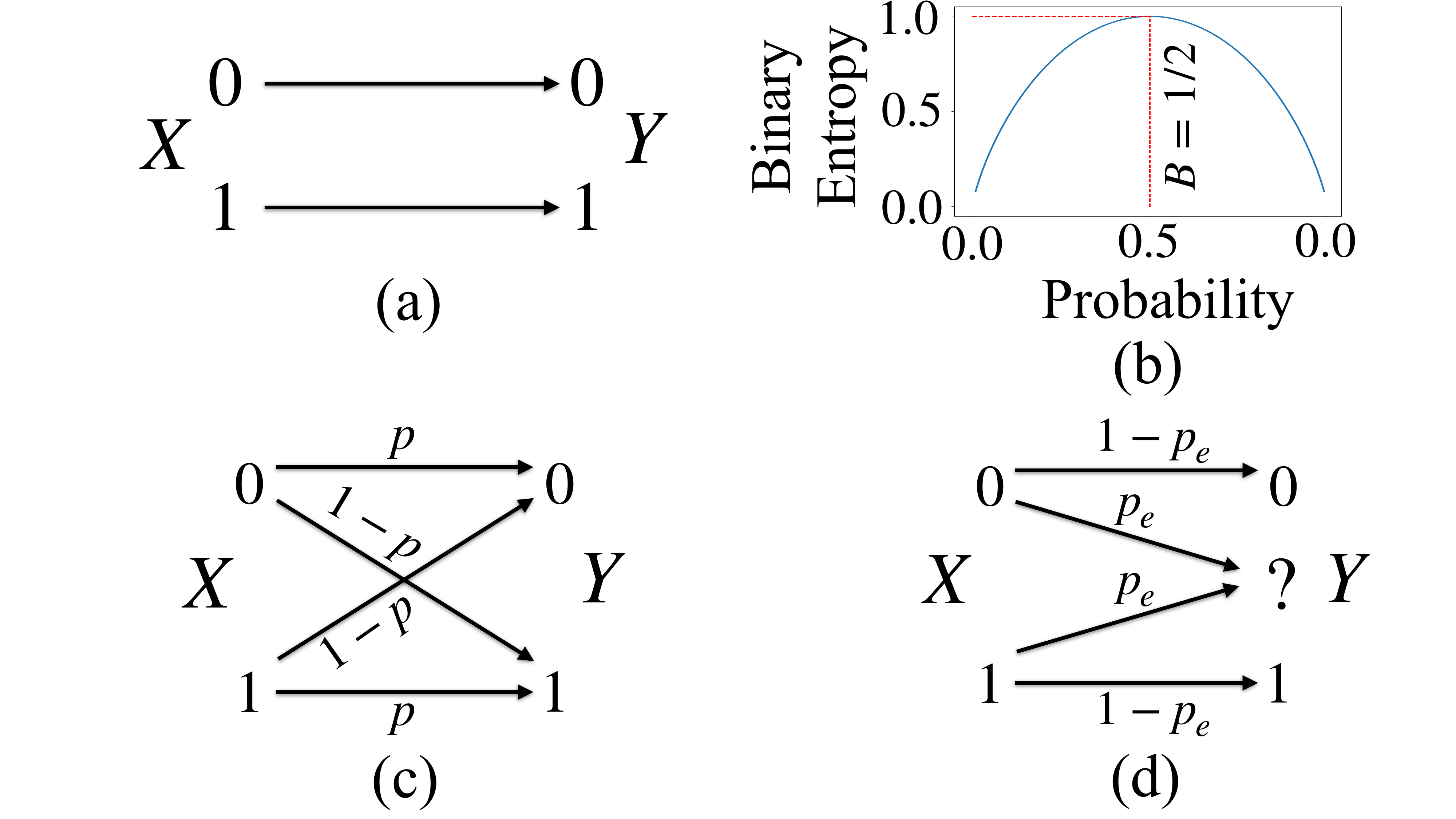}
    
    \caption{(a) Binary noiseless channel, (b) mutual information profile with respect to input distribution for noiseless channel (binary entropy function), (c) binary symmetric channel, (d) binary erasure channel }
    \label{classical_channels_all}
\end{figure}

However, if we assign energy values to the output and impose an energy constraint, then the capacity will decrease. For example, if the symbols ``0'' and ``1" carry 0 and 1 units of energy respectively, i.e.\ $b (y_1 = 0)=0$ and $b (y_2 = 1) = 1$. Then the capacity-power function is
\begin{equation*}
        C(B)=\begin{cases}
\log(2) & 0\leq B\leq\frac{1}{2}\\
{{h}(B)} & \frac{1}{2}\leq B\leq 1.
\end{cases}
    \end{equation*}
To illustrate this graphically, Fig.~\ref{classical_channels_all}(b) plots mutual information for the binary noiseless channel with respect to $P(1)$.  Clearly $\max_{P(1)} I(X;Y)$ is 1 and is achieved with $P(1)=1/2=P(0)$. The average output energy $E[b(Y)]$ corresponding to $\{P(0)=P(1)=1/2\}$ is $0\cdot\tfrac{1}{2} + 1\cdot\tfrac{1}{2} = \tfrac{1}{2}$.
 When the threshold output energy $B\leq 1/2$, the capacity-power function is just the unconstrained capacity since the capacity-achieving distribution already satisfies the minimum energy constraint. Now, let us consider beyond the region $0\leq B \leq 1/2$.  Clearly, the maximum possible value of $B$ is $1$ (when the bit string containing only ``1'' values is transmitted). So for $1/2<B\leq 1$, we are in the right half of Fig.~\ref{classical_channels_all}(b). Given a  value of $B$, it fixes $P(1) = B$, and any $P(1)\geq B$ will also satisfy the minimum energy constraint, but mutual information is maximized at $P(1) = B$ as the function is monotonically decreasing in this region.

Similarly, it follows the capacity-power function of the classical memoryless binary erasure channel with erasure probability $p_e$ (Fig.~\ref{classical_channels_all}(d)) is:
\begin{equation*}
        C(B)=\begin{cases}
(1-p_{e})\log(2) & 0\leq B\leq\frac{1}{2}\\
(1-p_{e}){h(B)} & \frac{1}{2}\leq B\leq1
\end{cases}
\end{equation*}
and for binary symmetric channel with bit flip probability $p$ (Fig.~\ref{classical_channels_all}(c)) is:
\begin{equation*}
        C(B)=\begin{cases}
\log(2) - h(p) & 0\leq B\leq\frac{1}{2}\\
h(B) - h(p) & \frac{1}{2}\leq B\leq1, p\leq 1-B.
\end{cases}
\end{equation*}
To study quantum analogues of these settings, we must consider quantum channels and quantum states. There are two aspects of communicating through quantum channels: sending classical information through quantum channels or sending quantum information through quantum channels. Here, we discuss the capacity-power function for transporting classical information through a quantum channel. In this protocol, we have a classical source $\mathbf{X}$, which produces some classical signal; we assign a quantum state to each letter and transport it through the channel.

The protocol for sending classical information through a quantum channel encodes information in quantum states at the transmitter.  After sending through the quantum channel, the receiver performs measurement on the states to get the best possible information about the classical input distribution; this is the \emph{accessible information}. If the transmitter  prepares $\rho_{x}$, where $x=0,1 \ldots, m$ with probabilities $p_0, p_1, \ldots, p_m$. Therefore, the transmitter has the state $\rho = \sum_{x} p_x \rho_x $. The receiver makes measurements using positive operator-valued measure (POVM) elements $\{E_y\} = \{E_1, E_2, \ldots, E_k\}$. {The accessible information is determined by the Holevo bound \cite{preskilllecture} for a noiseless 
channel.}and is given by:
\begin{equation}
    I(X:Y) \leq S(\rho) - \sum_{x} p_x S(\rho_x) = \chi(\xi).
    \label{holevo_chi}
\end{equation} 
Here, the Holevo quantity $\chi$  provides an upper bound on the accessible information and depends on the input ensemble that the transmitter chooses ($\xi =\{p_x,\rho_x\}$).  In practice, the information gain is much less than the upper bound due to the difficulty in finding optimum POVMs.
 As per \eqref{holevo_chi},  the highest information rate to the receiver is $S(\rho)$; any channel noise would decrease the rate. If an ensemble of input states $\xi =\{ p_x,\rho_x\}$ is mapped by a channel to an output ensemble $\xi^{'} =\{ p_x,\mathcal{N}(\rho_x)\}$ that is measured by the receiver, then the information gain (accessible information) is upper-bounded by the Holevo quantity of the output ensemble $\chi(\xi^{'})$. The single-use classical capacity of a channel is defined by \begin{equation}
    C_{1}^{Q}(\mathcal{N}) = \max_{\xi}\: \chi(\xi^{'}),
    \label{signle use capa}
\end{equation} 
where the superscript $Q$ indicates the quantum channel.
Since this is a property of the channel, it can also be written as $\chi_{1}(\mathcal{N})$.  The channel capacity is defined as 
\begin{equation}
	C^{Q}(\mathcal{N}) = \underset{n\rightarrow \infty}{\lim} \: \frac{1}{n} {C_{n}^{Q}(\mathcal{N})}.
\end{equation}
For additive $\mathcal{N}$, this regularization is not needed.   This holds for the maximum rate at which classical information can be transmitted through a noisy quantum channel $\mathcal{N}$ if the input contains product states (i.e., entanglement is not allowed, also known as the product-state classical capacity) and the output is measured by joint measurement setting. In this setting, for the quantum noisy communication channel $\mathcal{N}$, the {channel} capacity can be expressed as 
\begin{align}
C^{Q}(\mathcal{N}) &= C_{1}^{Q}(\mathcal{N})  =\max_{\xi}\: \chi(\xi^{'}) \notag  \\
&=\max_{\text {all }\{p_x, \rho_{x}\}}\left[S\left(\mathcal{N}(\rho)\right)-\sum_{x} p_{x} S\left(\mathcal{N}(\rho_{x})\right)\right]
\end{align}
where $\rho = \sum_{x}p_{x}\rho_{x}$ is the averaged state at the input. This is a generalization of Holevo's theorem \eqref{holevo_chi} for noisy quantum channels. When we allow only inputs that are unentangled over the channel use, then by optimizing our measurement setting, we can at most get the capacity $S(\mathcal{N}(\rho))$ (for pure states), maximized over the input ensemble.

\section{Capacity-Power Function of a Quantum Channel}

In this section, we go beyond the classical capacity-power function and formulate a quantum analogue of information transmission under energy constraints. As with classical power-capacity, we require a certain threshold energy at the output. For this purpose, we
consider a Hamiltonian operator whose expected value with respect to the output quantum state gives us the desired energy constraint. 

Let the single-use capacity-power function $C^{Q}_{1}(B)$ be the maximum rate at which classical information can be transmitted through a single use of a noisy quantum channel $\mathcal{N}$ such that we have a minimum average output energy $E[b(y)] = Tr[H\mathcal{N}(\rho)]\geq B$.  In this setting, for the  noisy quantum channel $\mathcal{N}$, the  capacity-power function is:
\begin{equation}
C^{Q}_{1}(B)= \underset{\substack{\{p_{x},\rho_{x}\}: \\ Tr(H\mathcal{N}(\rho))\geq B}}{\max} \chi= \underset{\substack{\{p_{x},\rho_{x}\}: \\ Tr(H\mathcal{N}(\rho))\geq B}}{\max}\left[S\left(\mathcal{N}(\rho)\right)-\sum_{x} p_{x} S\left(\mathcal{N}(\rho_{x})\right)\right]
\label{single capa energy}
\end{equation}
where the maximization is over a set of $\{p_{i},\rho_{i}\}$ such that the output state obeys $Tr[H\mathcal{N}(\rho)] \geq B$.

The definition is very natural and intuitive. Let $F$ be the set of all input distributions $\{p_{i},\rho_{i}\}$. The idea is to get the optimal rate under the output power constraint $B$. Restricting to only $B$-admissible input distributions will reduce $F$ to a new set $\mathcal{F}$, where $\mathcal{F}\subseteq F$. To get the optimum rate we must maximize $\left[S\left(\mathcal{N}(\rho)\right)-\sum_{i} p_{i} \mathrm{~S}\left(\mathcal{N}(\rho_{i})\right)\right]$ over $\mathcal{F}$.

For a block of length $n$, we define 
\begin{align}
C^{Q}_{n}(B) =  \underset{\substack{\{p_{X^{n}},\rho_{X^{n}} \}: \\ Tr(H^{\otimes n} \mathcal{N}^{\otimes n}(\rho^{\otimes n}))\geq nB}}{\max} S\left[\mathcal{N}^{\otimes n}\left(\sum_{x_{1},x_{2}...x_{n}} p_{x_{1},x_{2},\ldots,x_{n}}\rho_{x_{1}x_{2}, \ldots,x_{n}} \right)\right] \notag \\
    - \sum_{x_{1},x_{2}, \ldots,x_{n}} p_{x_{1},x_{2}, \ldots,x_{n}}S\left[\mathcal{N}^{\otimes n}(\rho_{x_{1}x_{2}, \ldots,x_{n}} )\right]
    \label{nth capa energy}
\end{align}

The capacity-power function is then defined as:
\begin{equation}
	C^{Q}(B) = \underset{n\rightarrow \infty}{\lim} \: \frac{1}{n} C^{Q}_{n}(B) \mbox{.}
	\label{capa energy function}
\end{equation}
The Holevo quantity $\chi$ is a concave function in probabilities and a convex function in the signal states \cite{wilde_book}. {This makes the computation of power-capacity for any general quantum channel difficult. However, we can design quantum channels by fixing the input quantum states. This seems more convenient operationally as well, where the experimentalist will have access to a fixed set of quantum states
and have the freedom to change the classical probability distribution involved. Such channels are known as classical-quantum channels.}

 \section{Properties of the Capacity-Power Function}
The idea is to encode classical information in quantum states and send it through a noisy quantum channel. Then the receiver will perform measurements and try to recover the original classical information. The single-use capacity-power function  $C_{1}^{Q}(B)$ is clearly non-increasing in $B$ as larger $B$ further restricts the feasible set.

Consider the input $\xi =\{ p_x, \rho_x\}$.
	When we fix the signal states $\rho_x$, and have the freedom to change the probabilities $p_x$, such a configuration is called a classical-quantum (CQ) channel. 
\begin{theorem}
		  
   $C_{1}^{Q}(B)$ is a \textbf{concave} function of $B$ for a CQ channel 
\begin{equation}
        C_{1}^{Q}(\alpha_{1}B_{1} + \alpha_{2}B_{2}) \geq \alpha_{1}C_{1}^{Q}(B_{1}) + \alpha_{2}C_{1}^{Q}(B_{2}) \mbox{.}
    \end{equation}
 \label{th1}   
\end{theorem}
\begin{proof}
Let us take two classical information sources ($X^{1},X^{2}$) distributed according to $P^{1},P^{2}$. Suppose $\rho^{1} = \{p_{x}^{1},\rho_{x}^{1}\}$ achieves $C_{1}^{Q}(B_{1})$ and $\rho^{2} = \{p^{2}_{x},\rho^{2}_{x}\}$ achieves $C_{1}^{Q}(B_{2})$. These states are mapped to output states after applying the channel $\mathcal{N}$. 
$$\rho^{1} / \rho^{2} \xrightarrow[\text{}]{\mathcal{N}} \sigma^{1} / \sigma^{2}\mbox{,}$$
    where $\sigma^{k} = \sum p^{k}_{x}\mathcal{N}(\rho_{x}^{k})$ for $k =1,2$. Then from from  \eqref{single capa energy},
    \begin{equation}
        \left[\mathrm{~S}\left(\sigma^{k}\right)-\sum_{x} p^{k}_{x} \mathrm{~S}\left(\sigma_{x}^{k}\right)\right] = C_{1}^{Q}(B_{k}) \mbox{.}
        \label{capacity energy individual}
    \end{equation}
    From our initial assumption,  the output states satisfy 
    \begin{equation}
    Tr[H\sigma^{k}] \geq B_{k}
    \label{energy bound on individual states}
\end{equation}     for $k = 1,2$. 
Now let us take another source $X$ distributed according to the distribution $P = \sum_{k} \alpha_{k}P^{k}$. Therefore, the prepared density matrix will be $\rho = \sum_{x}p_{x}\rho_{x}$ where $p_{x} = \alpha_{1}p^{1}_{x} + \alpha_{2}p^{2}_{x}$.

	As the signal states are fixed for a CQ channel
	 i.e.\ $\rho^{1}_{x} = \rho^{2}_{x} = \rho_{x} \xrightarrow[\text{}]{\mathcal{N}} \sigma_{x}$ for all $x$. Therefore the average output energy will be, from inequality \eqref{energy bound on individual states},
	
    \begin{align*}
        E[b(Y)] = Tr[H\mathcal{N}(\rho)] 
        =\sum_{k=1,2} \alpha_{k}Tr[H\sigma^{k}] \\
        \geq \sum_{k=1,2} \alpha_{k}B_{k}\mbox{.}
    \end{align*}
    So the ensemble $\{p_{x},\rho_{x}\}$ is $\sum_{k=1,2} \alpha_{k}B_{k}$-admissible and by definition \begin{equation}
        \left[S\left(\sigma\right)-\sum_{x} p_{x} S\left(\sigma_{x}\right)\right] \leq C_{1}^{Q}\left(\sum_{k=1,2} \alpha_{k}B_{k}\right)\mbox{,}
        \label{ineq}
    \end{equation} 
    where $\sigma = \sum_{x} p_x \sigma_x$
    Now the left side of this inequality is also a concave function of input probabilities ($\{p_{x}\}$), which implies \begin{align*}
        \left[S\left(\sigma\right)-\sum_{x} p_{x} S\left(\sigma_{x}\right)\right] \geq \alpha_{1}\left[S\left(\sigma^{1}\right)-\sum_{x} p^{1}_{x} S\left(\sigma^{1}_{x}\right)\right] \nonumber\\ 
        + \alpha_{2}\left[S\left(\sigma^{2}\right)-\sum_{x} p^{2}_{x} S\left(\sigma^{2}_{x}\right)\right].
    \end{align*}
    From \eqref{capacity energy individual},
    \begin{equation}
        \left[S\left(\sigma\right)-\sum_{x} p_{x} S\left(\sigma_{x}\right)\right] \geq \alpha_{1}C_{1}^{Q}(B_{1}) + \alpha_{2}C_{1}^{Q}(B_{2}).
        \label{ineq2}
    \end{equation}
    We get our desired result by combining inequality \eqref{ineq} with inequality \eqref{ineq2}.
    \begin{equation*}
        C_{1}^{Q}(\alpha_{1}B_{1} + \alpha_{2}B_{2}) \geq \alpha_{1}C_{1}^{Q}(B_{1}) + \alpha_{2}C_{1}^{Q}(B_{2}).
    \end{equation*}
\end{proof}

As discussed previously, the theorem holds for a CQ channel. In general, the set of states maximizing the channel capacity for two different values of threshold energy will be different. Therefore, the concavity property of power-capacity functions does not hold in general (see Fig.~\ref{verification of non concavity}).

\begin{figure}
\centering
\begin{subfigure}
    \centering
    \includegraphics[scale = 0.3]{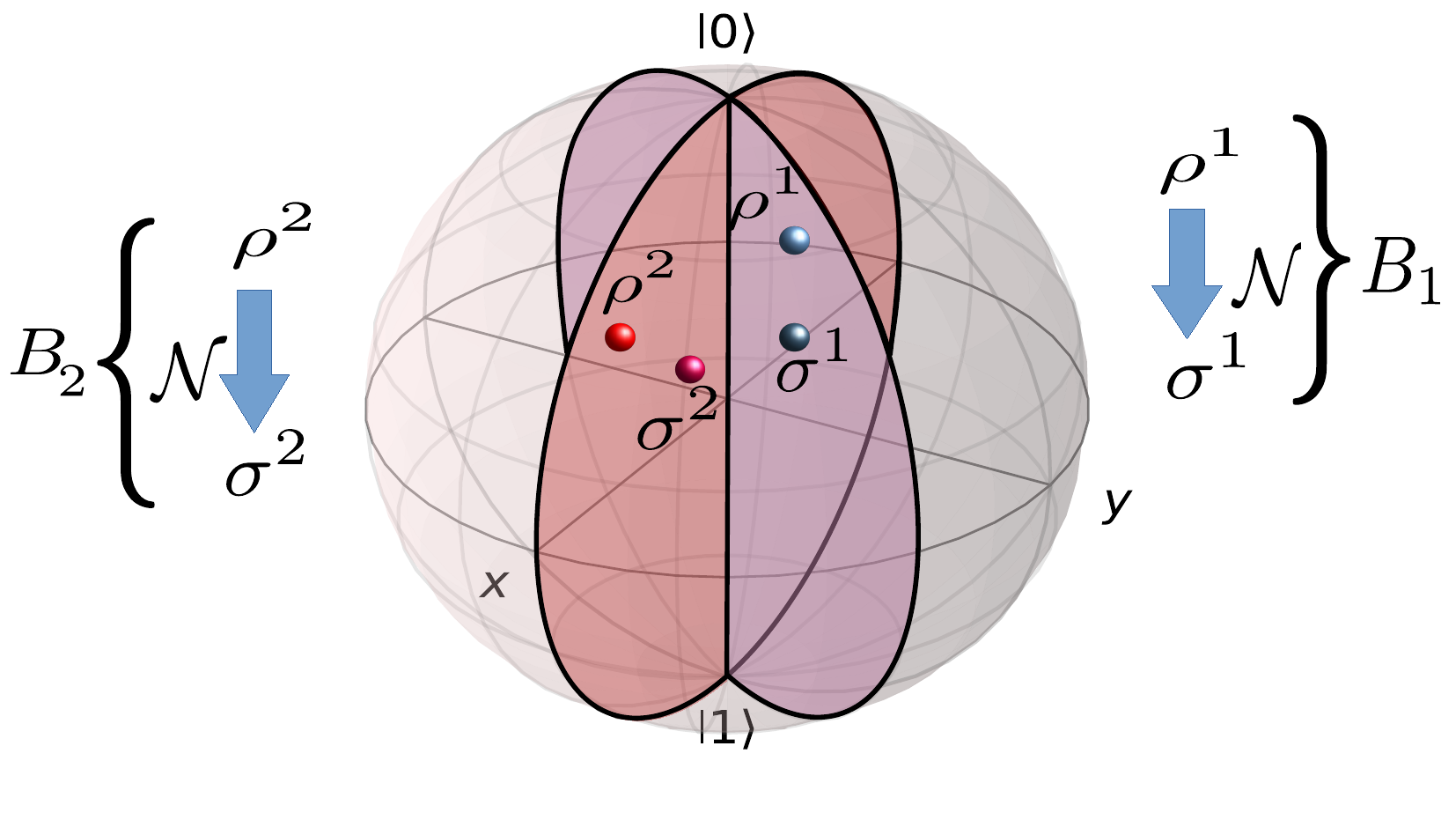}
        
    \caption{Visual representation of the states inside the Bloch sphere: $\rho^{1}$ and $\rho^{2}$ represent the input states which achieve capacity-power functions $C_{1}^{Q}(B_1)$ and $C_{1}^{Q}(B_2)$, respectively. Output of the channel is represented by $\sigma^{1}$ and $\sigma^{2}$, which, in general, lie in two different vertical planes.}
    \label{capa_en_bloch}
\end{subfigure}
\begin{subfigure}
    \centering
    \includegraphics[scale = 0.3]{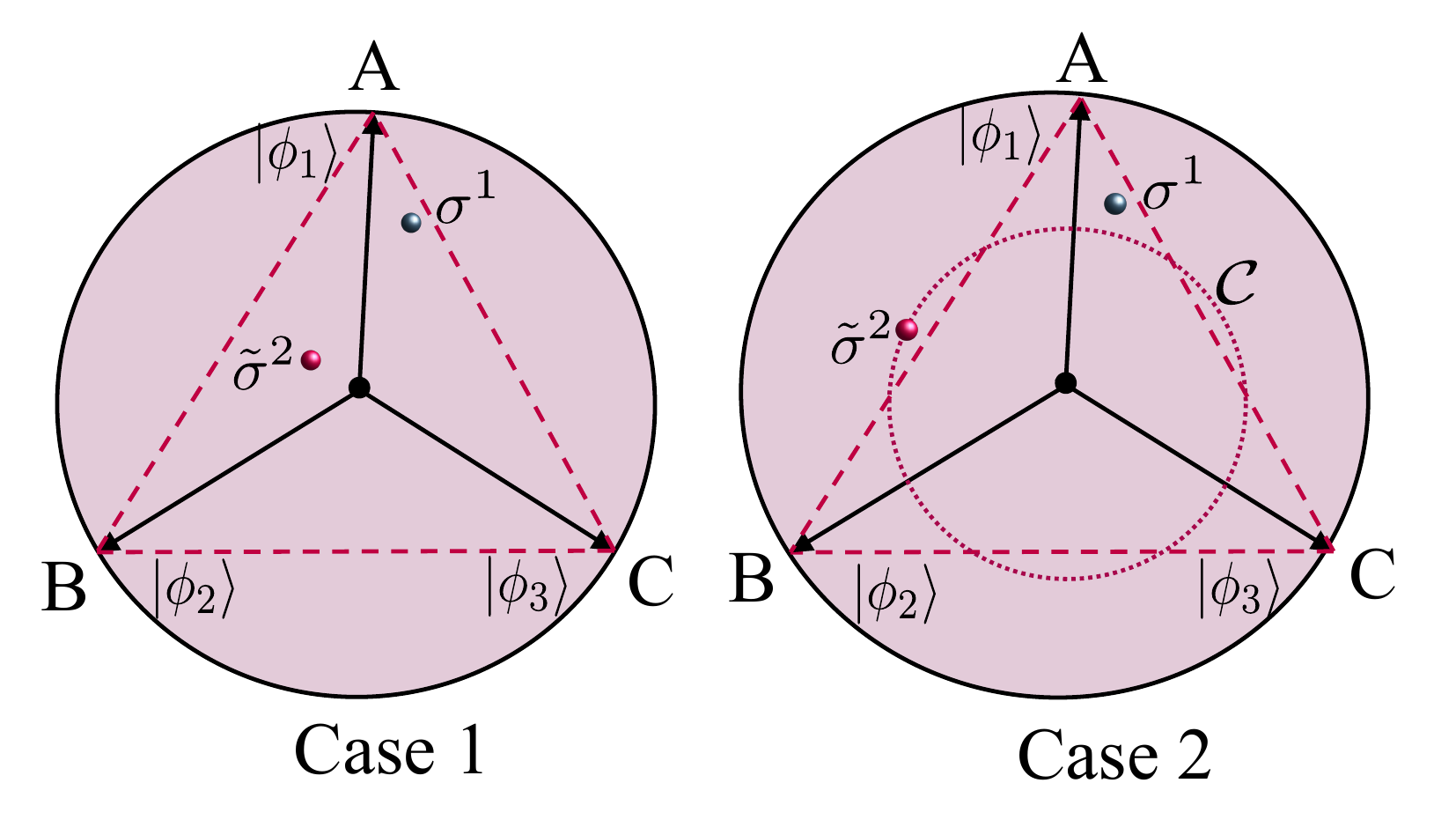}
   
    \caption{The plane where the output corresponding to energy $B_1$, $\sigma^1$ lies. Case 1: The state corresponding to energy $B_2$, after unitary transformation to a state on the same plane as $\sigma^1$, $\tilde{\sigma}^2$ lies outside of the triangle $ABC$. Case 2: The state corresponding to energy $B_2$, after unitary transformation to a state on the same plane as $\sigma^1$, $\sigma^2$ lies inside of the triangle $ABC$.}
    \label{capa_en_bloch_2}
\end{subfigure}
\end{figure}
\begin{figure}
    \centering
    \includegraphics[width=8.5cm, height=5.3cm]{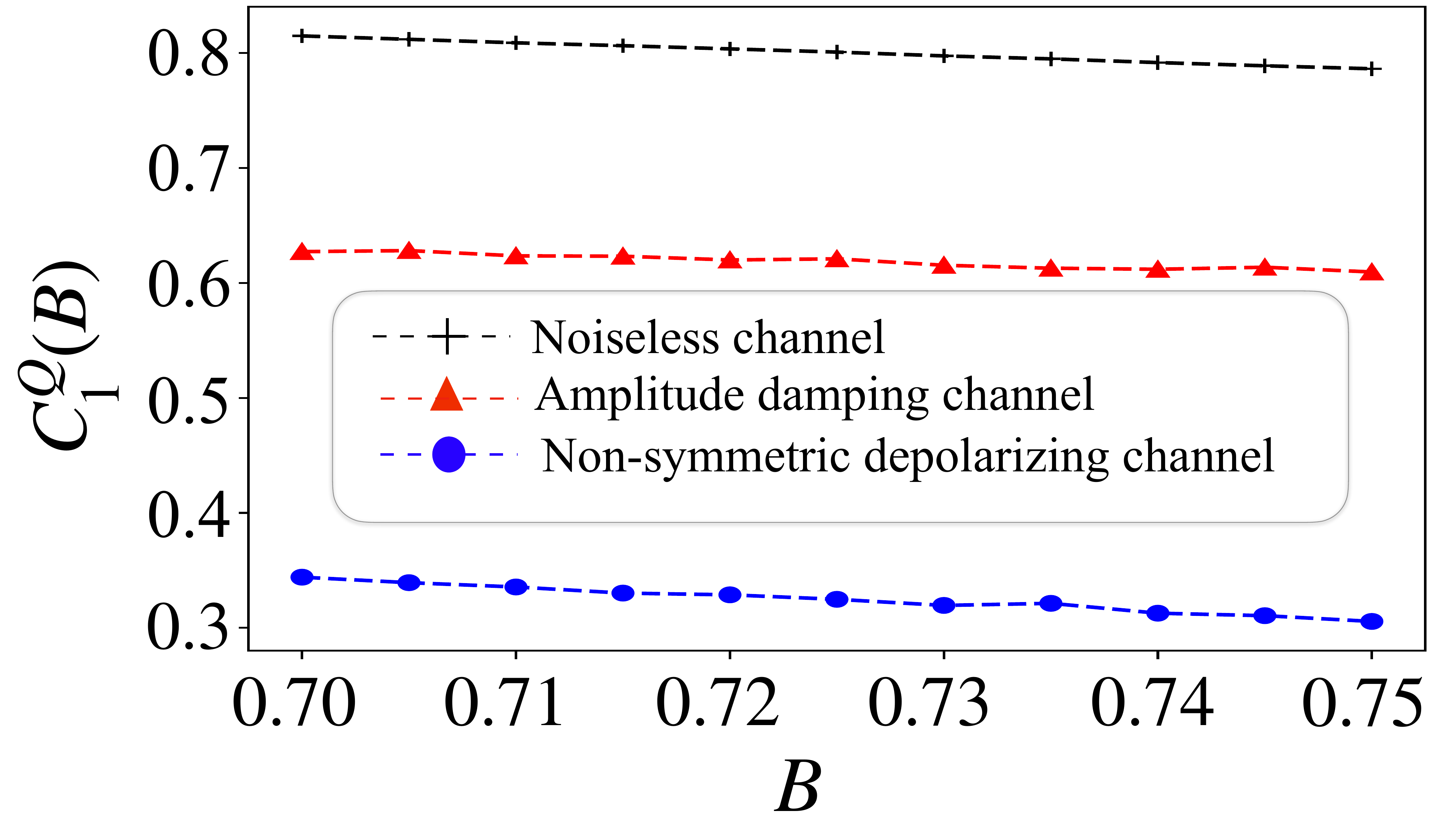}
    
    \caption{ Capacity-power function for three different channels when maximized over input states as well as probabilities. We can clearly see the function is not concave in general except for the noiseless channel (the black curve).}
    \label{verification of non concavity}
\end{figure}
\begin{corollary}
    $C_{1}^{Q}(B)$ is a \textbf{concave} function of $B$ for noiseless (or unitary) channel when maximized over pure states and probabilities as long as the set of signal states span the Hilbert space.
    \label{coro 1}
\end{corollary}

\begin{proof}
 Let us take two different input ensembles $\rho^1$ and $\rho^2$, transforming to $\sigma^1$ and $\sigma^2$ by a noiseless or a unitary channel, which achieves $C_{1}^{Q}(B_1)$ and $C_{1}^{Q}(B_2)$. In general, $\sigma^1$ and $\sigma^2$ will lie on two different vertical planes as shown by two different colors in Fig.~\ref{capa_en_bloch}. We can always transform the state $\sigma^2$ unitarily to a state $\tilde{\sigma}^2$ on the surface where $\sigma^1$ lies. The Holevo quantity is same for $\sigma^2$ and $\tilde{\sigma}^2$ which is $C_{1}^{Q}(B_2)$. In Fig.~\ref{capa_en_bloch_2} we can see $\sigma^1$ can be written as some convex combination of states $\{\ket{\phi_1},\ket{\phi_2},\ket{\phi_3}\}$. In fact, we can access any point within the triangle $ABC$ by some convex combination. Now, we can classify two different scenarios based on the position of $\tilde{\sigma}^2$, whether it is inside the triangle $ABC$ or not. 
 
Let us denote the first scenario as Case 1 and the second as Case 2 (see Fig.~\ref{capa_en_bloch_2}).
Case 1 is straightforward as we can reach $\tilde{\sigma}^2$ from $\sigma^1$ by only changing the probabilities. For Case 2, we can again transform the $\tilde{\sigma}^2$ along the circle $\mathcal{C}$ to reach inside the triangle $ABC$. Notice that this transformation is unitary and does not change the Holevo quantity of the ensemble. Once we reach inside, then we can follow the same procedure as in Case 1. 

 Therefore, there exists a distribution that achieves $C_1^{Q}(B_2)$ from the same set of input signal states which achieves $C_1^{Q}(B_1)$ by only changing the probabilities. We have already shown that the capacity-power function is concave under this scenario. We can now safely say $C_1^{Q}(B)$ is a concave function for any noiseless (or unitary) channel. 
\end{proof}

\begin{theorem}
    $C_{1}^{Q}(B)$ is a piecewise concave function for any noisy channel when maximized over both states and probabilities.
    \label{piece_wise_concavity}
\end{theorem}

\begin{proof}
For a given set of input signal states, the capacity-power function is concave (Theorem \ref{th1}).  Let us take three such sets of input signal states: $\mathcal{S}_1 = \{\rho_{1x} \}$ , $\mathcal{S}_2 = \{\rho_{2x} \}$ , $\mathcal{S}_3 = \{\rho_{3x} \}$. Between any two energy values $B_1$ and $B_2$, we will get three curves for these three sets of signal states, which are denoted by 1, 2, 3, respectively, in Fig.~\ref{graphical_proof_of_piecewise_concavity}(a). We can clearly see if we also maximize over these three sets of states, then the function will look like Fig.~\ref{graphical_proof_of_piecewise_concavity}(b). It will follow curve 1 till the point \emph{a} (crossing point between 1 and 2), then it will follow curve 2 between points \emph{a} and \emph{b} (crossing point between 2 and 3), and at the end, it will follow curve 3. In each interval, it is concave but not in the full domain. 
\end{proof}
\begin{figure}
        \centering
        \includegraphics[width=9cm, height=5cm]{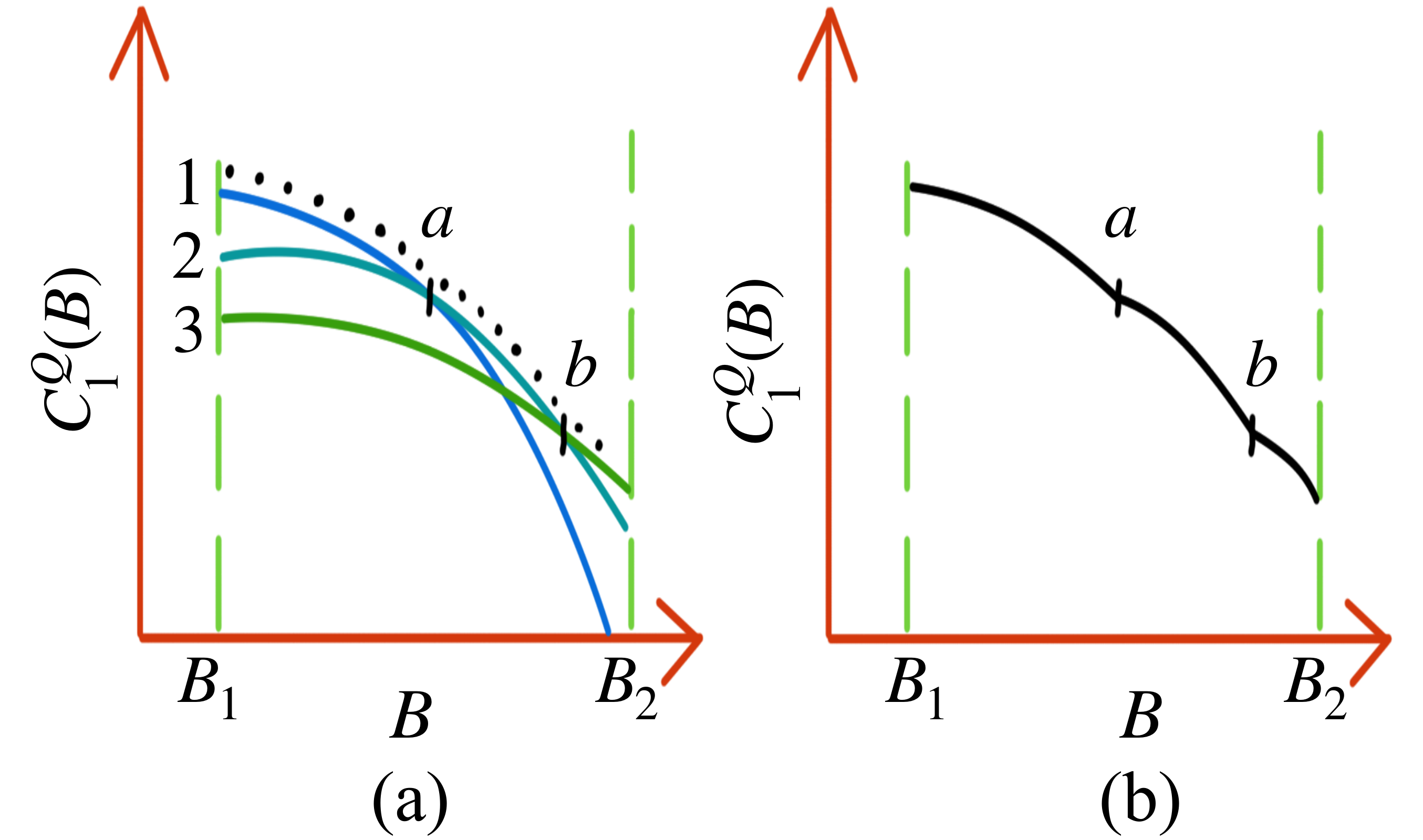}
      
      \caption{Visual proof of Theorem \ref{piece_wise_concavity}: (a) Three curves of capacity-power function corresponding to three different sets of input signal states, $\mathcal{S}_1 = \{\rho_{1x} \}$, $\mathcal{S}_2 = \{\rho_{2x} \}$, $\mathcal{S}_3 = \{\rho_{3x} \}$ denoted by 1, 2, 3, respectively. The maximum of capacity-power function over all these sets of input signal states, shown in (b), is piecewise concave.}
        \label{graphical_proof_of_piecewise_concavity}
    \end{figure}

Therefore, if we maximize over a large number of states, the distance between each crossing point will also decrease, eventually making the function almost concave.  Fig.~\ref{piecewise_concavity} illustrates this effect and emergence of an almost concave function upon increasing the number of states in our optimization protocol.

\begin{figure}
    \centering
    \includegraphics[width=8.5cm, height=5.5cm]{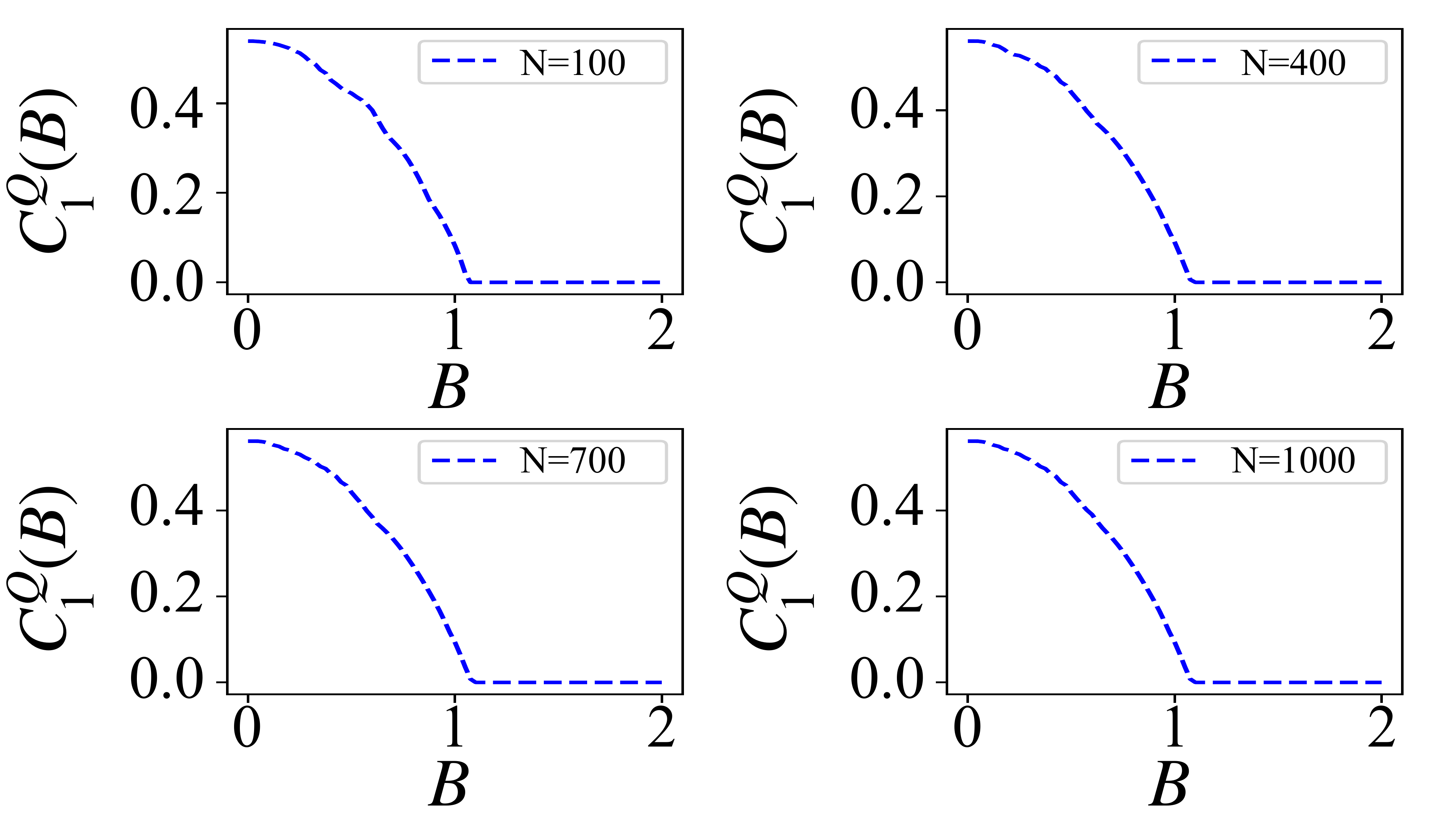}
    
    \caption{Capacity-power function for non-symmetric depolarizing channel maximized over $N$ number of states: we can see clearly the emergence of an ``almost" concave function as we increase $N$.}
    \label{piecewise_concavity}
\end{figure}

\begin{theorem}
     If we do not allow \textbf{entanglement or any classical correlation} between the inputs across the channel uses, i.e. the input ensemble is of the form \{$p_{x_{1}},\rho_{x_{1}}\} \otimes \{p_{x_{2}},\rho_{x_{2}}\} \ldots \otimes \{p_{x_{n}},\rho_{x_{n}}$\} for any CQ channel $\mathcal{N}^{\otimes n}$, where the signal states are fixed across the channel use and for any channel where capacity-power function is concave, the capacity-power function is additive, i.e.\ \begin{equation}
     C_{n}^{Q}(B) = nC_1^{Q}(B).
     \label{addivity}
 \end{equation}  
 \label{th2}
\end{theorem}
\begin{proof}
     Let us say $\xi^{*}=\{p_{x}^{*},\rho_{x}^{*}\}$ achieves $C_{1}^{Q}(B)$ for the channel $\mathcal{N}$. 
 \begin{equation}
    C_{1}^{Q}(B) = S\big[\mathcal{N} \big( \sum_{x}p_{x}^{*}\rho_{x}^{*}\big) \big] - \sum_{x}p_{x}^{*}S[\mathcal{N}(\rho_{x}^{*})] 
    \label{new}
\end{equation}
and
\begin{equation}
    Tr[H\mathcal{N} \big( \sum_{x}p_{x}^{*}\rho_{x}^{*}\big)] \geq  B.
    \label{new energy}
\end{equation}
 Now let us take the channel $\mathcal{N}^{\otimes 2}$ and the input ensemble of the form $\xi$ = \{$p_{x_{1}},\rho_{x_{1}}\} \otimes \{p_{x_{2}},\rho_{x_{2}}$\}. From definition 
\begin{align}
    C_{2}^{Q}(B) = \underset{\substack{\{p_{x_{1}}p_{x_{2}},\rho_{x_{1}} \otimes \rho_{x_{2}}\}: \\ Tr(H^{\otimes 2} \mathcal{N}^{\otimes2}(\rho^{\otimes 2}))\geq 2B}}{\max} S\left[\mathcal{N}^{\otimes 2}\left(\sum_{x_{1},x_{2}} p_{x_{1}}p_{x_{2}}(\rho_{x_{1}} \otimes \rho_{x_{2}})\right)\right] \notag \\ - \sum_{x_{1},x_{2}} p_{x_{1}}p_{x_{2}}S\left[\mathcal{N}^{\otimes 2}(\rho_{x_{1}} \otimes \rho_{x_{2}})\right]
    \label{c(B) two dim}
\end{align}
 The right side of \eqref{c(B) two dim} can be simplified.
 \begin{gather}
    \chi_{(\mathcal{N}^{\otimes 2})}(\xi) = S\left[\mathcal{N}^{\otimes 2}\left(\sum_{x_{1},x_{2}} p_{x_{1}}p_{x_{2}}(\rho_{x_{1}} \otimes \rho_{x_{2}})\right)\right] \notag
    \\ - \sum_{x_{1},x_{2}} p_{x_{1}}p_{x_{2}}S\big[\mathcal{N}^{\otimes 2}(\rho_{x_{1}} \otimes \rho_{x_{2}})\big]\\
    = S\left[\left(\sum_{x_{1}} p_{x_{1}}\mathcal{N}(\rho_{x_{1}})\right )\otimes \left(\sum_{x_{2}}p_{x_{2}}\mathcal{N}(\rho_{x_{2}})\right)\right] \notag
    \\ - \sum_{x_{1},x_{2}} p_{x_{1}}p_{x_{2}}S\left[\mathcal{N}(\rho_{x_{1}}) \otimes \mathcal{N}(\rho_{x_{2}})\right]\\
    = S\left[\sum_{x_{1}} p_{x_{1}}\mathcal{N}(\rho_{x_{1}})\right] +  S\left[\sum_{x_{2}}p_{x_{2}}\mathcal{N}(\rho_{x_{2}})\right] \notag
    \\ - \sum_{x_{1},x_{2}} p_{x_{1}}p_{x_{2}}\left(S\left[\mathcal{N}(\rho_{x_{1}})\right] + S\left[\mathcal{N}(\rho_{x_{2}})\right]\right ) 
\end{gather}
\begin{gather}
    = \underbrace{S\left[\sum_{x_{1}} p_{x_{1}}\mathcal{N}(\rho_{x_{1}})\right] - \sum_{x_{1}} p_{x_{1}}S\left[\mathcal{N}(\rho_{x_{1}})\right]}_{a}  \notag 
    \\ +  \underbrace{S\left[\sum_{x_{2}}p_{x_{2}}\mathcal{N}(\rho_{x_{2}})\right] - \sum_{x_{2}}p_{x_{2}}S\left[\mathcal{N}(\rho_{x_{2}})\right]}_{b}
    \label{decomposed_eq}
\end{gather}
We can see from \eqref{decomposed_eq} that $\chi_{\mathcal{N}^{\otimes 2}}(\xi)$ (Holevo quantity for the channel $\mathcal{N}^{\otimes 2}$) for the ensemble $\xi = \{p_{x_{1}}p_{x_{2}},\rho_{x_{1}} \otimes \rho_{x_{2}}\}$ is the sum of the individual $\chi_{\mathcal{N}}(\xi_{1})$ ($a$ in \eqref{decomposed_eq}) for the ensemble $\xi_{1} = \{p_{x_{1}},\rho_{x_{1}}\}$ and $\chi_{\mathcal{N}}(\xi_{2})$ ($b$ in \eqref{decomposed_eq}) for the ensemble $\xi_{2}= \{p_{x_{2}},\rho_{x_{2}}\}$. 

Now we must maximize information transmission over the given ensemble such that the output satisfies the energy bound. We  maximize $a$ and $b$ separately over $\xi_{1}$ and $\xi_{2}$ such that total average energy 
\begin{equation}
    Tr(H^{\otimes 2} \mathcal{N}^{\otimes2}(\rho^{\otimes 2}))=Tr(H\mathcal{N}(\rho_{1})) +  Tr(H\mathcal{N}(\rho_{2}))\geq 2B,  
    \label{energy_bound_2}
\end{equation}
where $\rho_{i} = \sum_{x_{i}}p_{x_{i}}\rho_{x_{i}}$. We can divide the total average energy as the sum of individual energies due to the absence of entanglement across multiple channel uses. One can take the ensemble $\xi^{*}$ for both $a$ and $b$ i.e.\ $\xi = \xi^{*}\otimes \xi^{*}$. It is straightforward to show this ensemble achieves the desired energy bound \eqref{energy_bound_2}. This means $\chi_{(\mathcal{N}^{\otimes 2})}(\xi^{*}) = 2C_{1}^{Q}(B)$. But there can be other input ensembles that can achieve the energy bound and give $a + b \geq 2C_{1}^{Q}(B)$. We can choose some ensemble $\xi_{1}$ and the corresponding state $\rho_1$ such that $Tr(H\mathcal{N}(\rho_{1})) \geq B_{1}$ and maximize over such ensembles which will lead us to $C_{1}^{Q}(B_{1})$ because of the monotone nature of $C_{1}^{Q}(B_{1})$, where $B_{1}\leq B$ implies $C_{1}^{Q}(B_{1})\geq C_{1}^{Q}(B)$. To achieve the desired energy bound, the energy from $\xi_{2}$ and the corresponding state $\rho_2$ should give the remaining energy. Therefore we maximize $b$ such that $Tr(H\mathcal{N}(\rho_{2})) \geq 2B - B_{1}$, which leads us to $C_{1}^{Q}(2B - B_{1})$. So we have $\max (a+b) = C_{1}^{Q}(B_{1}) + C_{1}^{Q}(2B - B_{1})$. 
\begin{align}
    x &= C_{1}^{Q}(B_{1}) + C_{1}^{Q}(2B - B_{1}) \\
    x/2 &= \frac{1}{2}(C_{1}^{Q}(B_{1}) + C_{1}^{Q}(2B - B_{1})) \\
    x/2 & \leq C_{1}^{Q}\left(\frac{1}{2} B_{1} + \frac{1}{2} (2B - B_{1})\right) 
    \label{concavity2}\\
    x &\leq 2C_{1}^{Q}(B)
    \label{ineqn}
\end{align}
for any $B$. We get the inequality in \eqref{concavity2} from the concave nature of $C_{1}^{Q}(B)$ from Theorem \ref{th1}. From \eqref{c(B) two dim}, \eqref{decomposed_eq}, \eqref{energy_bound_2}, and inequality \eqref{ineqn}, we can conclude \begin{equation}
    C_{2}^{Q}(B) = 2C_{1}^{Q}(B).
\end{equation}
Similarly, we can show the additivity \eqref{addivity} for any other value of $n$.
\end{proof}
    In the discussion above, we have considered input states that are separable. In the case of an entangled state, the channel capacity can be superadditive \cite{hastings2009superadditivity}. 
    Though not the central focus of our paper, we briefly discuss a specific case of entangled states with \textit{entanglement-breaking} channels and study the conditions (on the channel $\mathcal{N}$ and the input states) under which the additivity property for capacity-power function holds. [See Supplemental Material at URL will be inserted by publisher]
\subsection{Examples}

\subsubsection{A Noiseless Quantum Channel}
\label{sec:noiseless_ex}
\noindent Let us take an example from Preskill \cite{preskilllecture} and add energy restrictions to it. Let us take   $$|\phi_{1}\rangle = \begin{pmatrix}
    1 \\\
    0 
    \end{pmatrix}, |\phi_{2}\rangle = \begin{pmatrix}
    -\frac{1}{2} \\\
    \frac{\sqrt{3}}{2} 
    \end{pmatrix}, |\phi_{3}\rangle = \begin{pmatrix}
    -\frac{1}{2} \\\
    -\frac{\sqrt{3}}{2} 
    \end{pmatrix}$$ \\
    with probability $\frac{1}{3}$ for each. This is a spin $1/2$ object that points in one of the three directions symmetrically distributed in the x-z plane.
    The density matrix $$\rho = \frac{1}{3}\left(|\phi_1\rangle\langle\phi_1| + |\phi_2\rangle\langle\phi_2| + |\phi_3\rangle\langle\phi_3|\right) = \frac{\mathcal{I}}{2}.$$
    These states are to be sent through a quantum channel. The von Neumann entropy for this ensemble is $S(\rho) = 1$. So the Holevo bound tells us that we can have at most $1$ as accessible information. So if the channel is noiseless, then the capacity of this channel is $1$. Now we assign energies to the outputs: $b(1)=1/3,b(2)=2/3,b(3)=0$. Then the capacity-power function will be as in Fig.~\ref{fig:6}(b).

    But in practical use, the gained mutual information may be much less than $1$ as we cannot distinguish  states very well in one shot. Suppose we choose the POVMs corresponding to the outcomes $a=1,2,3$ as $E_{a} = \frac{2}{3}\left(\mathcal{I}-\ket{\phi_{a}}\bra{\phi_{a}}\right)$
    and so the crossover probabilities are  \begin{equation}
        p(a|b) = \bra{\phi_{b}}E_{a}\ket{\phi_{b}} = \begin{cases} 0 & a=b\\ \frac{1}{2} & a \neq b. \end{cases}
        \label{cross_prob}
    \end{equation}
Even if there is no noise in the quantum channel, we encounter crossover probabilities as in \eqref{cross_prob} and fall into the setting of a noisy classical channel.
\begin{figure}
    \centering
    \includegraphics[scale=0.5]{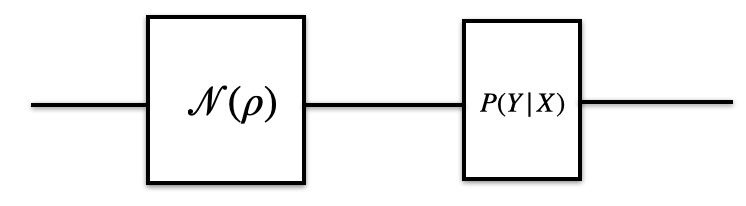}
    
    \caption{Decomposition of a typical quantum channel.}
    \label{fig:5}
\end{figure}    
Any quantum channel with non-orthogonal states as input and with any measurement setting can be decomposed as in Fig.~\ref{fig:5}. Therefore the mutual information between input and output  is $I(X:Y) = H(Y) - H(X|Y) = 0.5849$ (for this measurement setting), which is much less than the optimum value $1$. If we choose to encode the information to multiple inputs and use joint measurement settings at the output, then we can asymptotically reach the optimum rate.

\begin{figure}
    \centering
    \includegraphics[width = 8.5 cm]{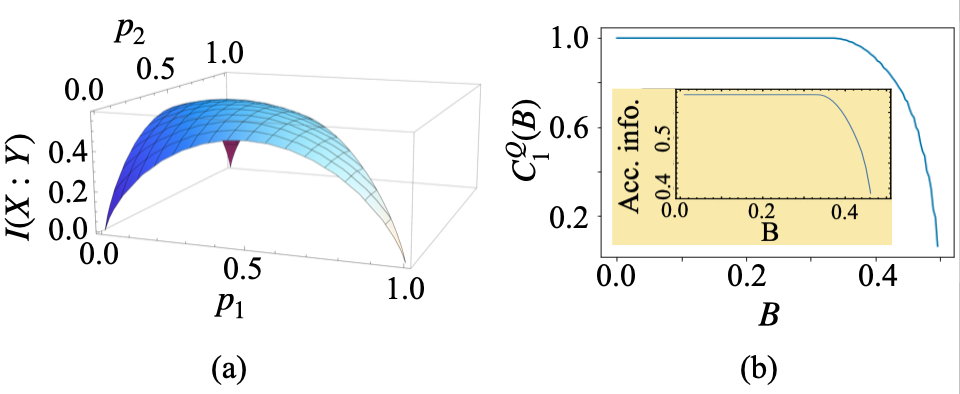}
    
    \caption{(a) Mutual information profile, where the x-axis is the input probability $p_{1}$ of choosing $\ket{\phi_{1}}$ and the y-axis is the input probability $p_{2}$ of choosing $\ket{\phi_{2}}$. The z-axis denotes the mutual information between the input and the output. (b) Capacity-power function for the noiseless channel of the example in Sec.~\ref{sec:noiseless_ex}; in the inset, the accessible information using the measurement setting described in Sec.~\ref{sec:noiseless_ex}.}
    \label{fig:6}
\end{figure}

If we plot mutual information as a function of input probabilities, we get Fig.~\ref{fig:6}(a). We can see that the maximum is roughly $0.5849$ and is achieved by $p_{1} = p_{2} = p_{3} = \frac{1}{3}$ as we saw from calculations before.

Note that this $I(X;Y)$ is not the actual capacity of the channel; the actual capacity of the channel is $S(\rho)=1$, which might be achieved by some set of POVMs. The capacity-power function (actually accessible information) for this particular measurement scheme is given by the plot in the inset of Fig.~\ref{fig:6}(b).

Let us discuss the physical origin of the plot. We assigned energies to the outputs: $b(1)=1/3$, $b(2)=2/3$, $b(3)=0$, then we can calculate the energy corresponding to the capacity-achieving distribution $\left(p_{1} = p_{2} = p_{3} = \frac{1}{3} \right)$ which is $\frac{1}{3}$. Clearly, if the desired threshold energy at the output, $B$, is less than $\frac{1}{3}$ then the capacity-power function is the same as the unconstrained capacity, but if $B\geq\frac{1}{3}$ then the allowed input probabilities will slide down along the surface in Fig.~\ref{fig:6}(a) yielding the capacity-power function for this measurement scheme.

\subsubsection{50-50 Beam Splitting Quantum Channel}
Now let us calculate the capacity-power function for a noisy quantum channel. Let $\rho_{in}^{s}$ be the state we want to send over the quantum channel $\mathcal{N}$. $\rho_{out}$ = $U\left(\rho_{in}^{s}\otimes \rho_{in}^{e}\right)U^{\dagger}$ and $\rho_{out}^{s} = Tr_{e}(\rho_{out})$.  Here, $$\rho_{out}^{s} = \mathcal{N}(\rho_{in}^{s})$$ is a completely positive trace-preserving (CPTP) map from input to output.

In the example, we take our input state $\rho_{in}^{s}$ to be in any of the pure coherent states $\{\ket{\alpha},\ket{\beta},\ket{\gamma}\}$ depending on the source probabilities. We have taken $U$ as a quantum beam splitter (BS), which acts with a probability $p_{b}$.
\begin{equation}
    U = \exp{\left[i\theta\left(a_{0}^{\dagger}a_{1} + a_{0}a_{1}^{\dagger}\right)\right]}
    \label{eq no.15}
\end{equation} 
where $a_{0}$ and $a_{1}$ are the annihilation operators of the two input ports. We have taken one of the inputs (state of the environment) of the BS as a vacuum state ($\ket{0}$) and $\theta= \pi/4$ (50-50 BS). So $$\rho_{in}^{e} = \ket{0}_{e}\bra{0},$$ at each use of the channel. We know that \cite{gerry_knight_2004}:
\begin{equation}
    \ket{\alpha}\ket{0}\xrightarrow[]{\text{BS}}\ket{\frac{\alpha}{\sqrt{2}}}\ket{\frac{\alpha}{\sqrt{2}}}.
\end{equation}
If the BS acts with probability $p_{b}$, then the possible outputs ($\rho_{out}^{s}$) of the channel are given in Fig.~\ref{fig:9}. 

\begin{figure}
    \centering
    \includegraphics[scale = 0.4]{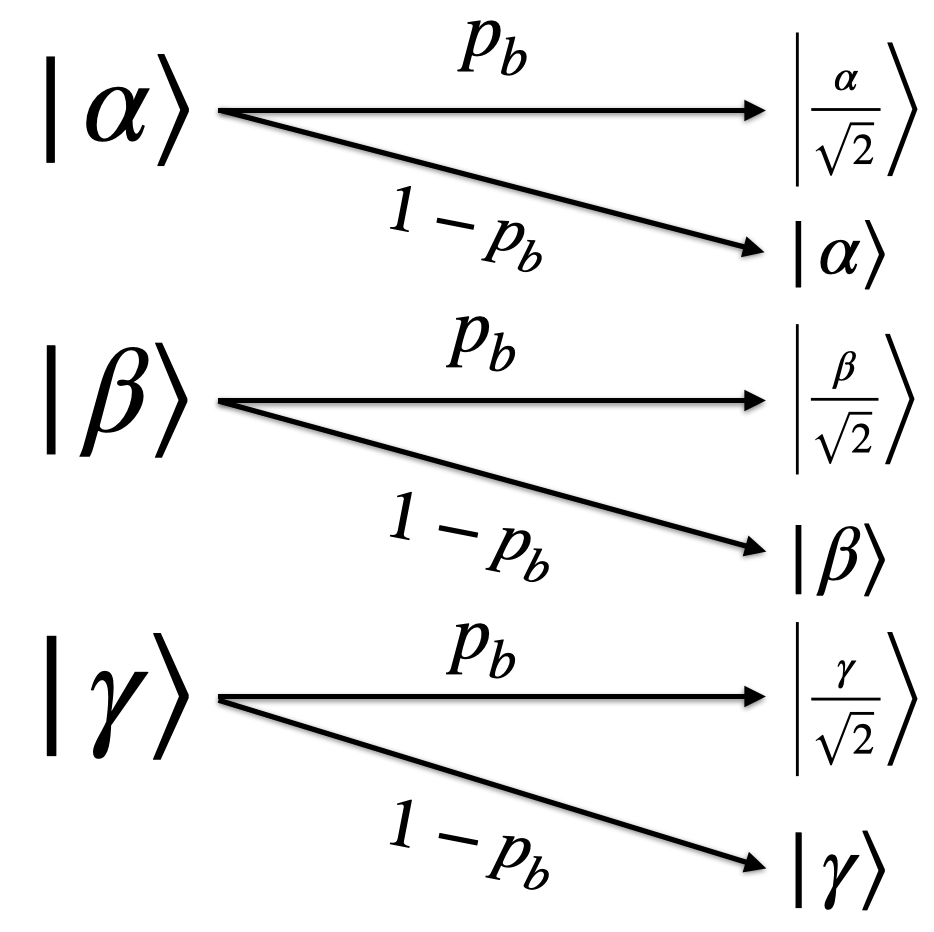}
    
    \caption{Possible outputs of the beam splitting channel}
    \label{fig:9}
\end{figure}
Our input state is $$\rho_{in}^{s} = \sum_{i=\alpha,\beta,\gamma}p_{i}\ket{i}\bra{i}.$$ It is follows directly that
$$\rho_{out}^{s} = \sum_{i=\alpha,\beta,\gamma}p_{i}p_{b}\ket{\frac{i}{\sqrt{2}}}\bra{\frac{i}{\sqrt{2}}} + \sum_{i = \alpha,\beta,\gamma}p_{i}(1-p_{b})\ket{i}\bra{i}.$$
Note that the capacity of any channel is monotone with respect to the noise probability $p_{b}$ because noise reduces the distinguishability of the states. 
We take $H = a^{\dagger}a$ which is common in practice; then each of the output states has an intensity of $|i|^{2}$ where $i = \alpha,\beta,\gamma,\alpha/\sqrt{2},\beta/\sqrt{2},\gamma/\sqrt{2}$.
For any given $p_{b}$ we get the capacity-power function and the average energy at the output will be $$Tr(H \mathcal{N}(\rho_{in}))= \sum_{i=\alpha,\beta,\gamma}p_{i}p_{b}|i|^{2}/2 + \sum_{i = \alpha,\beta,\gamma}p_{i}(1-p_{b})|i|^{2}.$$

The capacity-power function for this channel with three states in an ensemble for different values of noise parameter $p_{b}$ is shown in Fig.~\ref{capa_en_all_plots}(1.(a)--1.(c)). There are protocols (one such one is given in \cite{sidhu2021unambiguous}) to distinguish between coherent states using particular encoding and decoding schemes but all of them are bounded above by our results.
\begin{figure*}
\centering
\includegraphics[scale=.26]{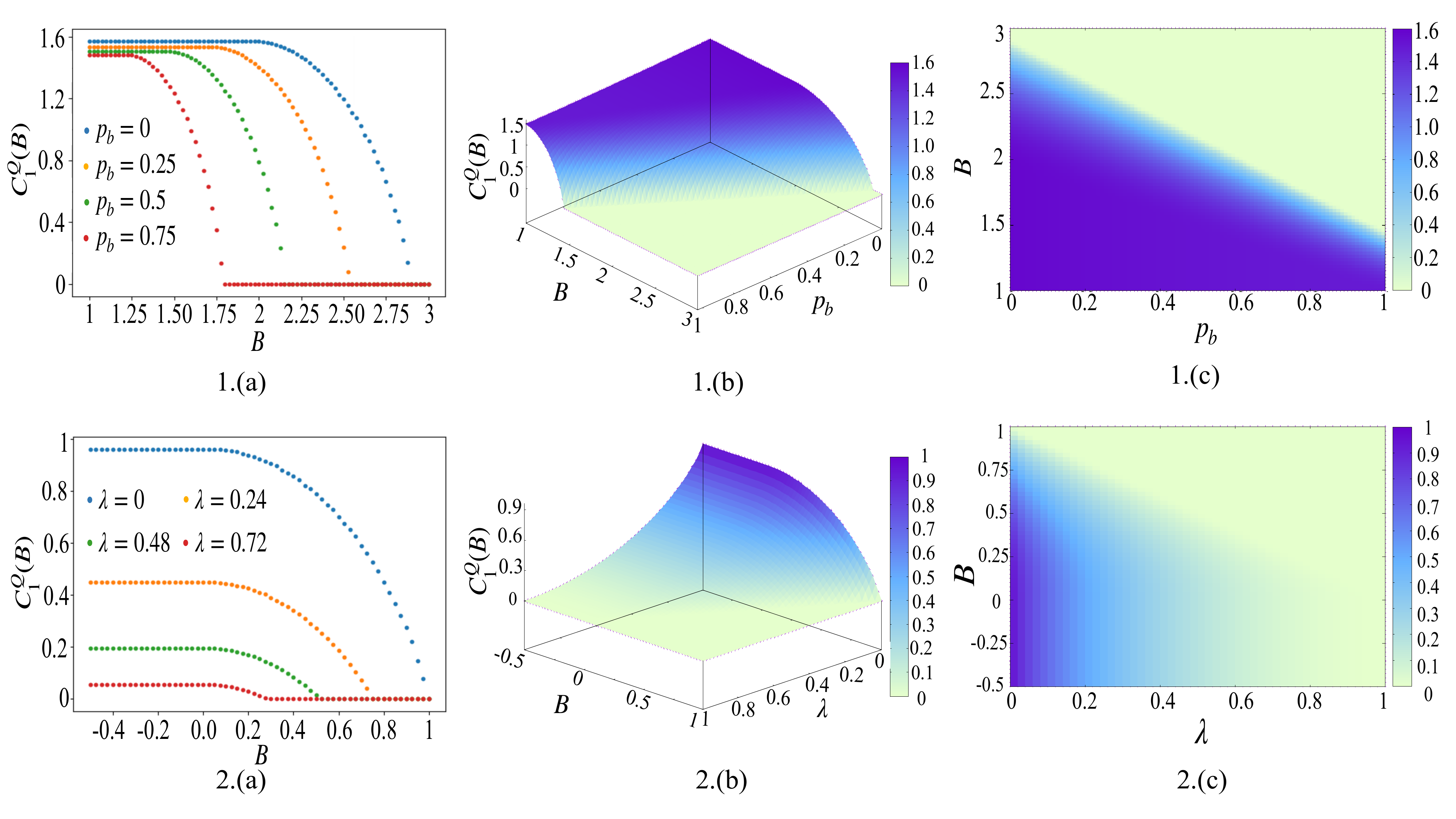} 

\caption{1.(a) Capacity-power function maximized over input probabilities for $\alpha = 1,\beta = \sqrt{2}e^{2\pi i/3},\gamma = \sqrt{3}e^{4\pi i/3}$ for different values of  noise parameter $p_{b}$ using \eqref{single capa energy}. 1.(b) Variation of capacity-power function for all values of the noise parameter $p_{b}$. 1.(c) Two-dimension projection of capacity-power function in $p_{b}-B$ plane. 2.(a) Variation of capacity-power function $C_{1}^{Q}(B)$ for the depolarizing channel maximized over the probabilities for different values of noise parameter $\lambda $. We took our signal states to be fixed at $\ket{\psi_{1}} = \ket{0}$, $\ket{\psi_{2}} = \cos(\pi/3)\ket{0} + \sin(\pi/3)\ket{1}$ and $\ket{\psi_{3}} = \cos(2\pi/3)\ket{0} + \sin(2\pi/3)\ket{1}$ 2.(b) Variation of the capacity-power function for the depolarizing channel over all noise parameters $\lambda$. 2.(c) Two-dimensional projection of capacity-power function for the depolarizing channel.
}

\label{capa_en_all_plots}
\end{figure*}

\subsubsection{Depolarizing Channel}

We calculate the capacity-power function for the quantum depolarizing channel with noise parameter $\lambda$.  The action of this channel on a state is given by 
$$\mathcal{N}_{\lambda}(\rho)=(1-\lambda) \rho+\frac{\lambda}{d} \mathds{1} \label{depol_channel}$$ 
where $d$ is the dimension of the Hilbert space.  We require $0 \leq \lambda \leq 1+\frac{1}{d^2-1}$ to make this channel completely positive.  To calculate the capacity-power function, we encode the message to one of three pure states $\{\ket{\psi_{i}}\}$, where $\ket{\psi_{i}} = \cos(\theta_{i}/2)\ket{0} + e^{i\phi_{i}}\sin(\theta_{i}/2)\ket{1}$, on the Bloch sphere where each state is chosen with probability $\{p_{i}\}$.  So \begin{align}
\rho_{in} =& \sum_{i}p_{i}\ket{\psi_{i}}\bra{\psi_{i}}\\
\rho_{out} =& \mathcal{N}_{\lambda}(\rho_{in}) = \sum_{i}p_{i} \mathcal{N}_{\lambda}(\ket{\psi_{i}}\bra{\psi_{i}}).
\end{align}
We take the Hamiltonian to be $\sigma_{z}$, which is relevant in the presence of an external magnetic field acting on the system in the $z$ direction. So from \eqref{single capa energy}, 
\begin{align}
C_{1}^{Q}(B) = \underset{\substack{\{p_{i},\ket{\psi_{i}}\}: \\ Tr(\sigma_{z}\rho_{out})\geq B}}{\max}\left[\mathrm{~S}\left(\mathcal{N}_{\lambda}(\rho_{in})\right)-\sum_{i} p_{i} \mathrm{~S}\left(\mathcal{N}_{\lambda}(\ket{\psi_{i}}\bra{\psi_{i}})\right)\right].
\end{align}
The capacity-power function for the depolarizing channel is shown in Fig.~\ref{capa_en_all_plots}(2.(a)--2.(c)).

We have seen that the capacity-power function is a concave function of the upper bound associated with the constraint, and this property also makes the function additive for a particular setting of the communication protocol. Now, suppose we have $\mathcal{R}$ number of constraints, i.e.\ the output state has to satisfy $\mathcal{R}$ inequalities. Let us take the constraints as $Tr(H_{1}\mathcal{N}(\rho_{in})\geq B_1, 
 Tr(H_{2}\mathcal{N}(\rho_{in})\geq B_2, 
 Tr(H_{3}\mathcal{N}(\rho_{in})\geq B_3, \ldots, Tr(H_{\mathcal{R}}\mathcal{N}(\rho_{in})\geq B_{\mathcal{R}}$.
If there exists a set of states that satisfy all the constraints, then the capacity-power function will be a concave function of all the upper bounds of the constraints and it will also be additive. 
\label{multiple_cons}
    
\section{Capacity-Power Function and Random Quantum States}

We are interested in the capacity-power function 
associated with random quantum states.
Consider the following scenario: The sender encodes information in a random pure state given by $\ket{\psi} = \sum_{i} \sqrt{p_{i}}e^{i\phi_{i}} \ket{i}$.
The receiver then makes projective measurements to obtain outcomes $\ket{i}$ with probability $p_i$. The density matrix after the action of the receiver is 
$\rho_{in} = \sum_{i = 1}^{N}p_{i}\ket{i}\bra{i}$. Thus, from an information-theoretic point of view, we can map a randomly chosen mixed state from an ensemble $\zeta$ to a random pure state $\ket{\psi}$.
Alternatively, one can consider a fiducial state acted upon by a  random unitary or a chaotic map.   If we take the Hilbert space dimension $N$ to be large, we can incorporate more orthogonal signal states in the input ensemble.

For an ensemble with pure input states that are orthogonal rank 1 projectors $\zeta = \{p_{i},\pi_{i}\}$, the noiseless channel capacity is given by the von Neumann entropy of the input ensemble, which is also the Shannon entropy of the probabilities.  From random matrix theory we know if we take any uniformly distributed random pure state  $\ket{\psi} = \sum_{i} \sqrt{p_{i}}e^{i\phi_{i}} \ket{i}$, then the joint distribution of the $\{p_{i}\}$ will be uniform over the $N-1$ dimensional probability space \cite{wooters}. To get this uniform distribution we need to draw each $p_{i}$ from a distribution \begin{equation}
	\rho(p_{i}) = N e^{-Np_{i}}.
	\label{weight}
\end{equation}
The normalization condition of the $\{p_{i}\}$, i.e.\ that $\sum p_{i} = 1$, is automatically satisfied for large $N$ as $\langle\sum p_{i}\rangle = 1$ and $var(\sum p_{i}) \propto N^{-1}$. So a typical random state from such an ensemble will have an entropy \cite{wooters,Zyczkowski_1990}: 
\begin{equation}
	\mathcal{H}=-N^{2} \int_0^{\infty}p\ln p \:  e^{-Np} d p=\ln N-0.422784.
	\label{typical_entropy}
\end{equation}
The constant $0.422784$ is $1 - \gamma$, where $\gamma$ is the Euler-Mascheroni constant. We give another way to arrive at the above  expression [See Supplemental Material at URL will be inserted by publisher] which essentially was first given in \cite{madhok2008}.
The variance of $\mathcal{H}$ goes to zero for large $N$ (see Fig.~\ref{variance}), akin to typicality arguments in information theory that rely on the concentration of measure. The entropy $\mathcal{H}$ for a typical random state  only differs from the maximum possible entropy $\ln N$ by a constant. Hence, we do not need a maximization process to compute capacity corresponding to a  large Hilbert space, since any randomly picked state has the same typical entropy. 
Therefore,  capacity of a noiseless channel corresponding to  a large dimension $N$ 
attains the value given by \ref{typical_entropy}. This is also the channel capacity associated with typical states with no energy constraint.

\begin{figure}
	\centering
	\includegraphics[scale = 0.13]{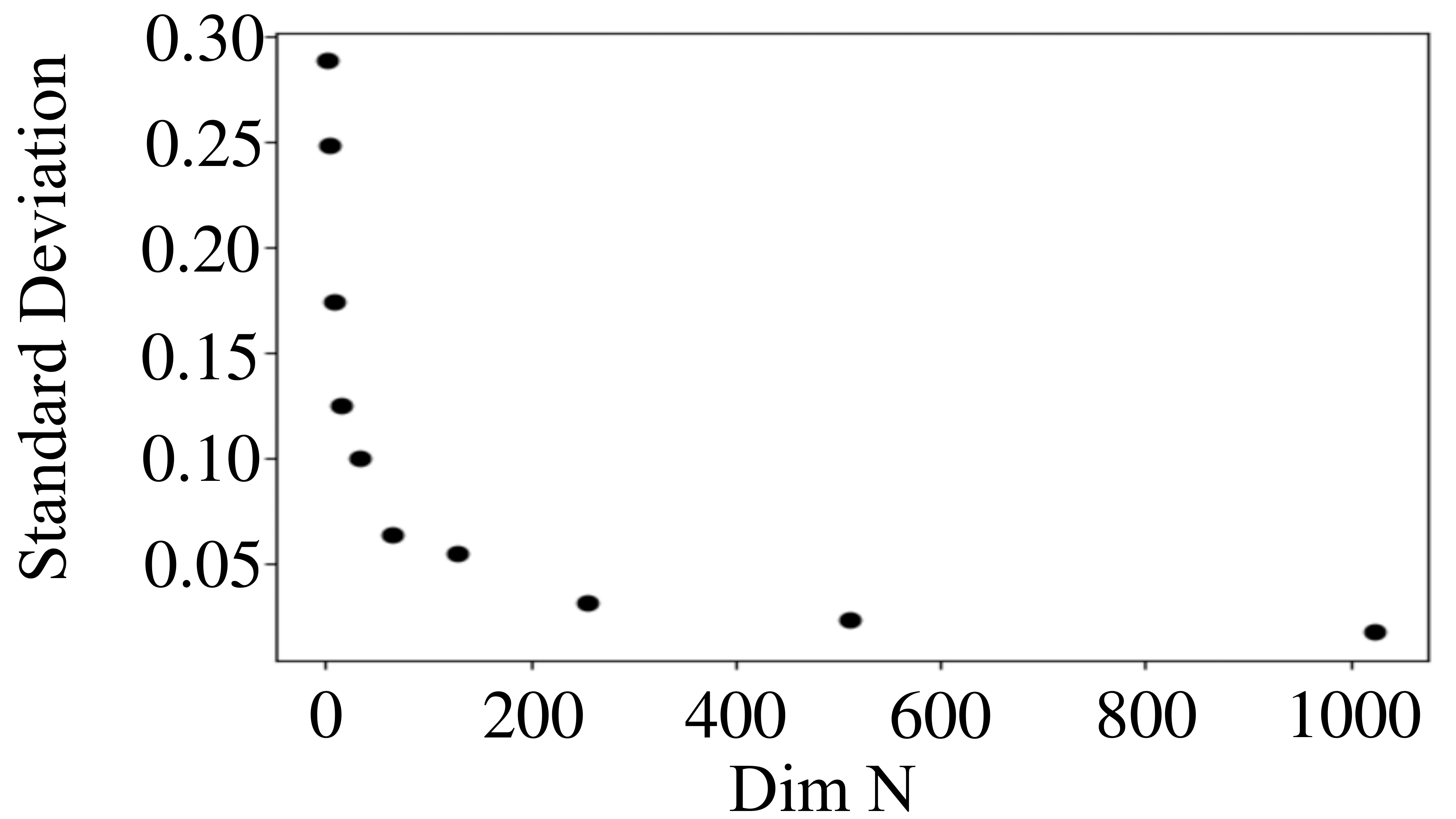}
	
	\caption{Standard deviation of noiseless channel capacity $\mathcal{H}$ as a function of the dimension of Hilbert space.}
	\label{variance}
\end{figure}

 In order to get the capacity-power function for a given threshold of output energy, 
 we require $Tr[H\rho] \geq B$. In other words, 
 $\sum_{i}p_{i}b_{i}\geq B$ where $b_{i} = Tr[H\pi_i] $. In this case, we need to randomly pick states which are uniform, not all over the sphere but from the part consistent with the above constraint. As the distribution of the probabilities is uniform over the $N-1$ dimensional probability sphere for uniformly distributed random pure state, we can expect the same for this sliced space of states. If we pick the probabilities from \begin{equation}
	\rho\left(p_n\right)=\left(\nu+\mu b_n\right) \exp \left[-\left(\nu+\mu b_n\right) p_n\right]
\end{equation} along with the conditions \begin{align}
	\sum_n \frac{1}{\nu+\mu b_n}&=1 \label{56} \\
	\sum_n \frac{b_n}{\nu+\mu b_n}&=B \label{57}
\end{align}
we can get the desired uniformity over the sliced space of probabilities \cite{wooters} that satisfy  $\sum_{i}p_{i}b_{i}= B$. The typical entropy associated with these conditions is 
\begin{align}
	\mathcal{H} &=-\sum_n\left(\nu+\mu b_n\right) \int_0^{\infty}\left(p_n \ln p_n\right) e^{-\left(\nu+\mu b_n\right) p_n} d p_n \\
	&=\tilde{\mathcal{H}}-0.422784 \label{entropy_energy}
\end{align} where $\tilde{\mathcal{H}}$ is the entropy of the probability distribution defined by $P_{n} = 1/(\nu + \mu b_{n})$.

Let us denote  $B_{t}$ as the value of the energy associated with a typical state. 

\begin{equation}
	B_{t}  = \langle B \rangle = \langle\sum_{n} p_{n}b_{n}\rangle = \sum \langle p_{n} \rangle b_{n} = \sum_{n} \frac{b_{n}}{N}
\end{equation} and the variance about $B_{t}$ is 
\begin{equation}
	\langle B^{2} \rangle - \langle B \rangle ^{2} =  \frac{\sum b_{n}^{2}}{N^{2}}\rightarrow 0 \:\:\text{as} \:\:N\rightarrow \infty.
\end{equation}
Here we use the fact that $\langle p_{n} \rangle = 1/N$ and $\langle p_{n}p_{m} \rangle -  \langle p_{n}\rangle \langle p_{m} \rangle = \delta_{mn}/N^{2}$ for uniform distribution (from \eqref{weight}).

When the threshold output energy $B\leq B_t$, the capacity-power function is just the unconstrained capacity since the capacity-achieving distribution already satisfies the minimum energy constraint. Moreover, for sufficiently large dimensions, the capacity becomes independent of the distribution (typical entropy) as can be seen in \eqref{typical_entropy}. 

Beyond $B_{t}$, $C_{1}^{Q}(B)$ is a monotonically  decreasing function of $B$. 
for $B > B_{t}$, the capacity-power function will be given by \eqref{entropy_energy} where $\nu$ and $\mu$ can be determined from \eqref{56} and \eqref{57}.

 Therefore, we can write the analytical expression for the capacity-power function for a noiseless channel for a system of large dimensions by only using the statistical properties without the computational maximization step. The capacity-power function can be written as 
\begin{align}
	C_{1}^{Q}(B) = 
	\begin{cases}
		\ln{N} - 0.422784 & 0 \leq B\leq B_{t}  \\
		- \sum_{n}\frac{1}{\nu + \mu b_{n}}\ln\left(\frac{1}{\nu + \mu b_{n}}\right) - 0.422784 & \forall B> B_{t} 
	\end{cases}  \label{capa_en_wooters}
\end{align}
where $\nu$ and $\mu$ are from \eqref{56}, \eqref{57}. The constant $0.422784$ will not affect the expression significantly for large dimensions.

It follows directly that $C_{1}^{Q}(B)$ from \eqref{capa_en_wooters} is a concave function of $B$ which matches with Theorem \ref{th1}. To prove this let us consider two different distributions of probabilities labelled as $\rho(p_{n}^1)$ and $\rho(p_{n}^2)$ parameterized by $\{\nu_1,\mu_1\}$ and $\{\nu_2,\mu_2\}$ respectively where the parameters obey:
\begin{align}
	\sum_n \frac{1}{\nu_1+\mu_1 b_n}&=1 \:\:\text{and}\:\: \sum_n \frac{1}{\nu_2+\mu_2 b_n}=1 \label{B1_1} \\
	\sum_n \frac{b_n}{\nu_1+\mu_1 b_n}&=B_1 \:\:\text{and}\:\:\sum_n \frac{b_n}{\nu_2+\mu_2 b_n}=B_2 \label{B2_2}
\end{align}
Therefore $C_1^{Q}(B_1)/C_1^{Q}(B_2) = \tilde{\mathcal{H}_1}-0.422784 / \tilde{\mathcal{H}_2}-0.422784$ where $\tilde{\mathcal{H}_k}$ is the entropy of the probability distribution defined by $p_{n}^k =\frac {1}{(\nu_k + \mu_k b_{n})}$. Note that we are only interested in the region $B_k\geq B_{t}$.

Now take another variable $p_n$ denoted as $\alpha_1 p_n^{1} + \alpha_2 p_n^{2}$. It is clear from \eqref{B2_2} that the energy corresponding to this probability distribution will be $B = \alpha_1 B_1 + \alpha_2 B_2$. We have already seen that the capacity-power function is nothing but the average Shannon entropy of the distribution. This implies 
\begin{align}
   &  \:\:\: C_1^{Q}(B) = \left< -\sum p_n \ln{p_n}\right> \\
    & = \left< -\sum (\alpha_1 p_n^{1} + \alpha_2 p_n^{2}) \ln{(\alpha_1 p_n^{1} + \alpha_2 p_n^{2})}\right> \\
    &\geq \alpha_1\left< -\sum  p_n^{1} \ln{ p_n^{1} }\right>  + \alpha_2  \left<  -\sum p_n^{2} \ln{  p_n^{2})}\right>
\end{align}
Therefore $C_1^{Q}(B) \geq \alpha_1 C_1^{Q}(B_1)+\alpha_2 C_1^{Q}(B_2)$.

\section{Private Capacity-Power of a Quantum Channel}

The private capacity of a quantum channel is the maximum rate at which classical information can be sent \emph{reliably} and \emph{privately}; operational definitions and coding theorems are given in prior work  \cite{energy_constrained_private_and_quantum_capacity,Devetak2005,Devetak_private}. The protocol is shown in Fig.~\ref{fig: Private channel}. 
Along with channel capacity, private capacity has been extensively studied \cite{wilde_book}. It is natural to extend how this quantity behaves under 
energy constraints and properties of private capacity-power functions. 
As the reliability of information transmission is a highly desirable property, 
one may ask the question: How does reliable transfer of information behave when the goal is to also transmit energy to the receiver as well?


\subsection{Private Capacity-Power Function of a Quantum Channel}
\begin{figure}
    \centering
    \includegraphics[scale = .25]{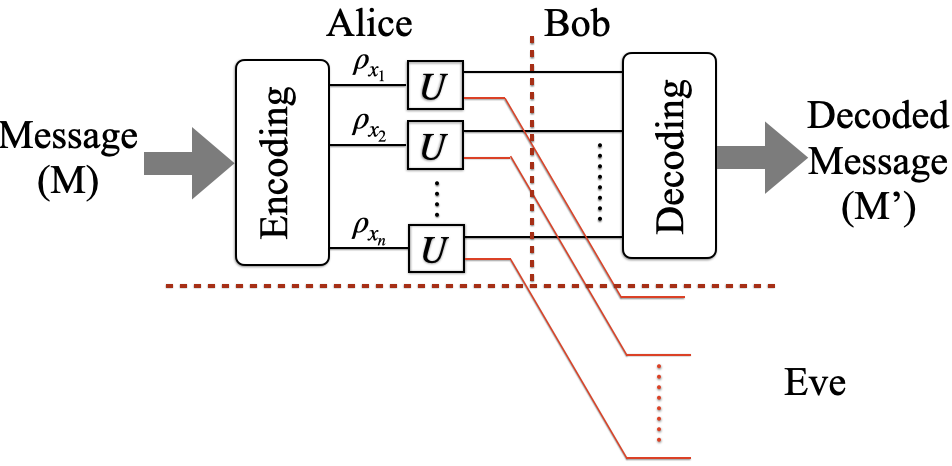}
    
    \caption{A typical protocol for communicating privately}
    \label{fig: Private channel}
\end{figure}

The single-use private capacity of channel $\mathcal{N}$ is defined as \begin{equation}
P_{1}^Q(\mathcal{N}) = \max _{\{p_{x}, \rho_{x}\}}\left(\chi_{A B}^{(1)}-\chi_{A E}^{(1)}\right)
\end{equation} where $\chi^{(1)}_{AB}$ is the same as $\chi(\xi^{'})$ of \eqref{signle use capa} and $\chi^{(1)}_{AE}$ is the Holevo quantity of the output of the complimentary channel $\mathcal{N}^{c}$.
The  regularized private capacity is defined as 
\begin{equation}
    P^{Q}(\mathcal{N}) = \underset{n\rightarrow \infty}{\lim} \frac{1}{n}P_{n}^{Q}(\mathcal{N}^{\otimes n}).
    \label{regulasied private capacity}
\end{equation}
Equivalently, 
\begin{equation}
P^{Q}(\mathcal{N})= \underset{n\rightarrow \infty}{\lim}\: \frac{1}{n}\max _{\{p_{x}, \rho_{x}\}}\left(\chi_{A B}^{(n)}-\chi_{A E}^{(n)}\right).
\label{private capa in terms of holevo chi's}
\end{equation}
We define the private capacity-power function $P(B)$ for a quantum channel as
\begin{equation}
P^{Q}(B)=\underset{n\rightarrow \infty}{\lim} \frac{1}{n} \left [\underset{\substack{\{p_{X^{n}},\rho_{X^{n}} \}: \\ Tr(H^{\otimes n} \mathcal{N}^{\otimes n}(\rho^{\otimes n}))\geq nB}}{\max}\left(\chi_{A B}^{(n)}-\chi_{A E}^{(n)}\right)\right]
\label{private capa energy function in terms of holevo chi's}
\end{equation} 
The operational significance of this definition follows directly from previous coding theorems on private capacity.

\subsection{Properties of Private Capacity-Power Function}

For an ensemble $\{p_{i},\rho_{i}\}$, let us denote $\chi_{A B}^{(n)}-\chi_{A E}^{(n)}$ as  $\mathcal{X}^{(n)}_{P}(\{p_{i},\rho_{i}\})$. 

\begin{theorem}
    $\mathcal{X}^{(1)}_{P}(\{p_{i},\rho_{i}\})$ is a \textbf{concave} function of input probabilities for degradable channels.
\begin{equation}
	\mathcal{X}^{(1)}_{P}(\{\alpha_{1}p_{i}^{1} + \alpha_{2}p_{i}^{2},\rho_{i}\})  \geq \alpha_{1} 	\mathcal{X}^{(1)}_{P}(\{p_{i}^{1},\rho_{i}\}) + \alpha_{2} 	\mathcal{X}^{(1)}_{P}(\{p_{i}^{2},\rho_{i}\})
	\label{concavity of private chi}
\end{equation}
\label{th3}
\end{theorem} 
\begin{proof}
    From the definition of the Holevo quantity in \eqref{holevo_chi} we can write
\begin{align*}
	\mathcal{X}^{(1)}_{P}(\{\alpha_{1}p_{i}^{1} + \alpha_{2}p_{i}^{2},\rho_{i}\}) = \chi_{A B}^{(1)}-\chi_{A E}^{(1)}  
\end{align*}
\begin{align}
	 &=\underbrace{S\left[\mathcal{N}\left(\sum_{i}\left(\alpha_{1}p_{i}^{1} + \alpha_{2}p_{i}^{2}\right)\rho_{i}\right)\right]}_{a} - \sum_{i}\left(\alpha_{1}p_{i}^{1} + \alpha_{2}p_{i}^{2}\right)S\left[\mathcal{N}\left(\rho_{i}\right)\right]  \notag  \\
&- \underbrace{S\left[\mathcal{N}^{c}\left(\sum_{i}\left(\alpha_{1}p_{i}^{1} + \alpha_{2}p_{i}^{2}\right)\rho_{i}\right)\right]}_{b} - \sum_{i}\left(\alpha_{1}p_{i}^{1} + \alpha_{2}p_{i}^{2}\right)S\left[\mathcal{N}^{c}\left(\rho_{i}\right)\right]
\label{eq38}
\end{align}
Now it is clear that concavity will come from parts $a$ and $b$ of \eqref{eq38} since the other two terms are linear.  So we must show 
\begin{align}
 a-b &\geq \alpha_{1}\left(S\left[\mathcal{N}\left(\sum_{i}p_{i}^{1}\rho_{i}\right)\right] - S\left[\mathcal{N}^{c}\left(\sum_{i}p_{i}^{1}\rho_{i}\right)\right]\right)  \notag  \\
 &+ \alpha_{2}\left(S\left[\mathcal{N}\left(\sum_{i}p_{i}^{2}\rho_{i}\right)\right] - S\left[\mathcal{N}^{c}\left(\sum_{i}p_{i}^{2}\rho_{i}\right)\right]\right). 
 \label{ineq 39}
\end{align}
From the properties of degradable channels \cite{wilde_book}, we know a channel $\mathcal{N}$ is degradable if its complimentary channel $\mathcal{N}^{c}$ can be written as 
$$\mathcal{N}^{c} \equiv \mathcal{D} \circ \mathcal{N}$$ where $\mathcal{N}^{c}$ is the channel from the system to the environment,  i.e.\ $\mathcal{N}^{c}(\rho) = Tr_{\{B\}} \:(U_{A\rightarrow BE}\rho_{A} U_{A\rightarrow BE}^{\dagger})$.

Let us write $\sigma^{1}=\mathcal{N}\left(\sum_{i}p_{i}^{1}\rho_{i}\right)$ and $\sigma^{2}=\mathcal{N}\left(\sum_{i}p_{i}^{2}\rho_{i}\right)$. So inequality \eqref{ineq 39} can be rewritten as 
\begin{align}
 a-b &\geq \alpha_{1}\left(S\left[\sigma^{1}\right] - S\left[\mathcal{D}\left(\sigma^{1}\right)\right]\right) + \alpha_{2}\left(S\left[\sigma^{2}\right] - S\left[\mathcal{D}\left(\sigma^{2}\right)\right]\right) 
 \label{eq40}
\end{align} where $a$ = $S\left[\alpha_{1}\sigma^{1} + \alpha_{2}\sigma^{2}\right]$ and $b$ = $S\left[\alpha_{1}\mathcal{D}\left(\sigma^{1}\right) + \alpha_{2}\mathcal{D}\left(\sigma^{2}\right)\right]$. The von Neumann entropy satisfies concavity,  i.e.\ for $\sum_{i}\lambda_{i} = 1$, 
\begin{equation}S\left(\sum_{i}\lambda_{i}\rho_{i}\right)\geq\sum_{i}\lambda_{i}S\left(\rho_{i}\right).
\label{concavity of S}
\end{equation} 

We can therefore write 
\begin{align}
a &= S\left[\sum_{k = 1,2}\alpha_{k}\sigma^{k}\right] \geq \sum_{k = 1,2}\alpha_{k}S\left[\sigma^{k}\right] \\
b &= S\left[\sum_{k = 1,2}\alpha_{k}\mathcal{D}(\sigma^{k})\right] \geq \sum_{k = 1,2}\alpha_{k}S\left[\mathcal{D}(\sigma^{k})\right]. 
\label{eq42}
\end{align} 
Now consider the quantity 
\begin{align}
x &= \left(S\left[\sum_{k = 1,2}\alpha_{k}\sigma^{k}\right] - \sum_{k = 1,2}\alpha_{k}S\left[\sigma^{k}\right]\right)\notag \\
&- \left(S\left[\sum_{k = 1,2}\alpha_{k}\mathcal{D}(\sigma^{k})\right] - \sum_{k = 1,2}\alpha_{k}S\left[\mathcal{D}(\sigma^{k})\right]\right)  \\
&= \chi(\Lambda) - \chi(\mathcal{D}(\Lambda))
\end{align} where ensemble $\Lambda = \{\alpha_{k},\sigma^{k}\}$.

From the data processing inequality we can say $x\geq0$ as $\chi(\Lambda) \geq \chi(\mathcal{D}(\Lambda))$, 
so 
\begin{align}
a-b = S\left[\sum_{k = 1,2}\alpha_{k}\sigma^{k}\right]  - S\left[\sum_{k = 1,2}\alpha_{k}\mathcal{D}\left(\sigma^{k}\right)\right]  \notag\\ 
\geq \sum_{k = 1,2}\alpha_{k}S\left[\sigma^{k}\right]  - \sum_{k = 1,2}\alpha_{k}S\left[\mathcal{D}\left(\sigma^{k}\right)\right]. \label{eq45}
\end{align}

Inequality \eqref{eq45} is the same as inequality \eqref{eq40}.  We had proven inequalty \eqref{eq40} and from \eqref{eq38} it is clear that inequalty \eqref{concavity of private chi} holds. 
\end{proof} 
\begin{corollary}
    $P_{1}^{Q}(B)$ is a concave function of B for degradable CQ channels.
\begin{equation}
    P_{1}^{Q}(\alpha_{1}B_{1} + \alpha_{2}B_{2}) \geq \alpha_{1}P_{1}^{Q}(B_{1}) + \alpha_{2} P_{1}^{Q}(B_{2})
    \label{concavity of P(B)}.
\end{equation}
\label{cor1}
\end{corollary}
\begin{proof}
    We can show the private capacity-power function is concave in $B$ by proceeding similarly as we did for the unassisted case in Theorem \ref{th1}, Corollary \ref{coro 1}, and using Theorem \ref{th3}.
\end{proof}  

\begin{theorem}
    $P_{1}^{Q}(B)$ is a piecewise concave function for any noisy degradable channel when maximised over both states and probabilities.
\end{theorem}
\begin{proof}
    Same argument as Theorem \ref{piece_wise_concavity}.
\end{proof}

\begin{theorem}
    $P_{n}^{Q}(B)$ is additive over channel uses if we only use product states without any correlation at the input for any degradable CQ channel and also for the channels for which the $P_{1}^{Q}(B)$ is concave 
\begin{equation}
	P_{n}^{Q}(B) = nP_{1}^{Q}(B).
\end{equation}
\label{th4}
\end{theorem} 
\begin{proof}
This additive property can be shown in a similar way as for unassisted capacity in Theorem \ref{th2} and making use of concavity from Corollary \ref{cor1}. 
\end{proof}

We see that for a degradable CQ channel, \eqref{private capa energy function in terms of holevo chi's} reduces to the single-letter formula:
\begin{equation}
    P^{Q}(B) =\max _{\{p_{i}\}: Tr(H\mathcal{N}(\rho))\geq B}\left(\chi_{A B}^{(1)}-\chi_{A E}^{(1)}\right).
    \label{private capa energy single letter formula}
\end{equation}

\subsection{Examples}
\subsubsection{Private Capacity-Power Function for Amplitude Damping Channel}
An amplitude damping channel is denoted by its Kraus operators $K_{0} = \begin{pmatrix}
    1 & 0\\
    0 & \sqrt{1-\lambda}
\end{pmatrix}$
and $K_{1} = \begin{pmatrix}
    0 & \sqrt{\lambda} \\
    0 & 0
\end{pmatrix}. $
Amplitude damping channel is a degradable channel for $\lambda \leq 0.5$. We consider an ensemble of 3 random pure qubit states with corresponding probabilities and maximized over all such ensembles. The density matrix of the environment is $$\rho_{out}^{e} = \begin{pmatrix}
    1-\lambda p & \sqrt{\lambda}\eta^{*} \\
    \sqrt{\lambda}\eta & \lambda p
\end{pmatrix}$$ when the input density operator is $$\rho_{in} = \begin{pmatrix}
    1-p & \eta^{*} \\
    \eta &  p
\end{pmatrix}.$$
We take the Hamiltonian to be $\sigma_{3}$ as we had done for calculating the unassisted capacity-power function for depolarizing channels. To calculate the private capacity-power function we take $\{\ket{\psi_{i}}\}$, where $\ket{\psi_{i}} = \cos(\theta_{i}/2)\ket{0} + e^{i\phi_{i}}\sin(\theta_{i}/2)\ket{1}$, on the Bloch sphere where each state is chosen with probability $\{p_{i}\}$, similar to what we did for the unassisted depolarizing channel. We consider an ensemble of 3 random pure qubit states as we took for unassisted case, with corresponding probabilities and maximized over all such probabilities.
We plot results in Figs.~\ref{fig:P(B)_for amp_damp} and \ref{fig:P(B)_for amp_damp_map}. Clearly, the private capacity-power function is concave in $B$. 
\begin{figure}
    \centering
    \includegraphics[scale = 0.13]{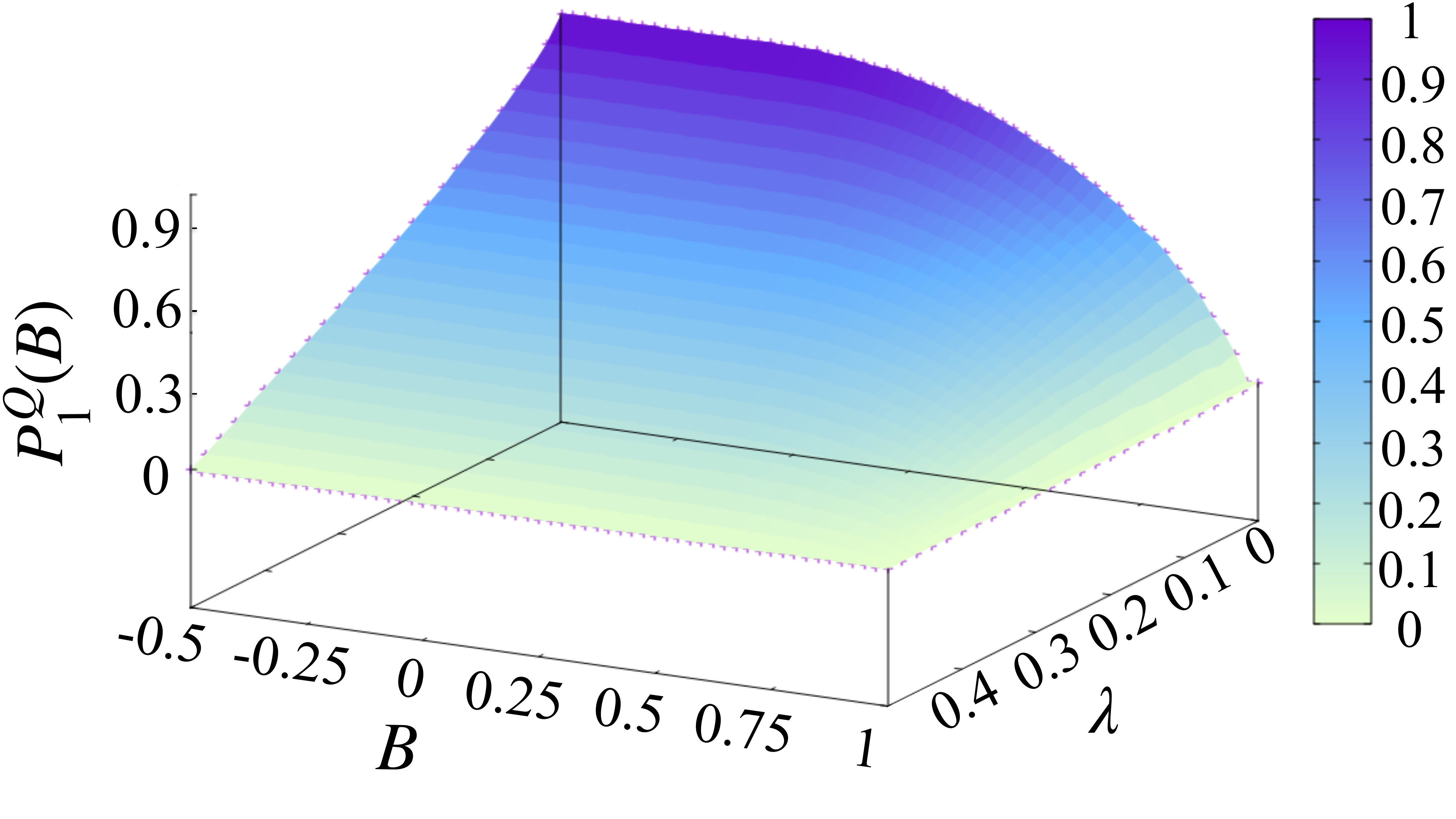}
   
    \caption{Variation of private channel capacity-power function $P_{1}^{Q}(B)$ for amplitude damping channel for different values of noise parameter $\lambda$.}
    \label{fig:P(B)_for amp_damp}
\end{figure}
\begin{figure}
    \centering
    \includegraphics[scale = 0.13]{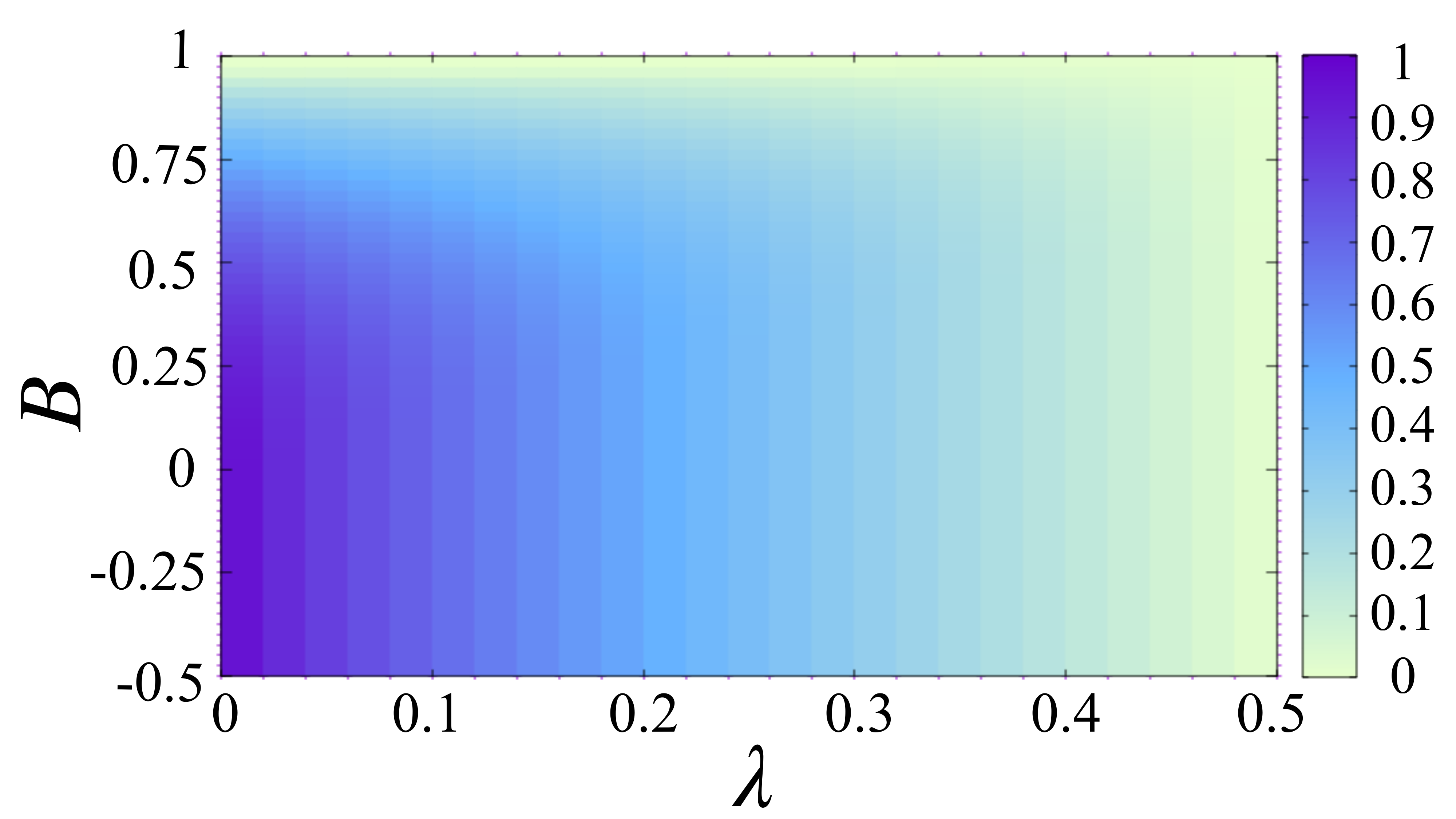}
    
    \caption{Projection of private channel capacity-power function $P_{1}^{Q}(B)$ for amplitude damping channel on $B$-$\lambda$ plane.}
    \label{fig:P(B)_for amp_damp_map}
\end{figure}
\subsubsection{Private Capacity-Power Function for Depolarizing Channel}
A depolarizing channel with noise parameter $\lambda$ takes an input state $\rho_{in}$ to the maximally mixed state with probability $\lambda$ and keeps it unaltered with probability $\lambda$. 
The isometric extension of this channel is given by 
\begin{align}
U_{A\rightarrow BE} &= \sqrt{1-\lambda} \: \mathds{1} \otimes \ket{0} \notag \\
&+ \sqrt{\lambda / 3}\left(\sigma_{1}\otimes\ket{1} + \sigma_{2}\otimes \ket{2}+\sigma_{3}\otimes \ket{3}\right).
\end{align}
To calculate the private capacity-power function of this channel, we consider an ensemble of 3 random pure qubit states as we took for unassisted case, with corresponding probabilities and maximized over all such probabilities. The results are plotted in Figs.~\ref{fig:P(B)_for _depol_3d} and \ref{fig:P(B)_for_depol_map}.
\begin{figure}
    \centering
    \includegraphics[scale = 0.13]{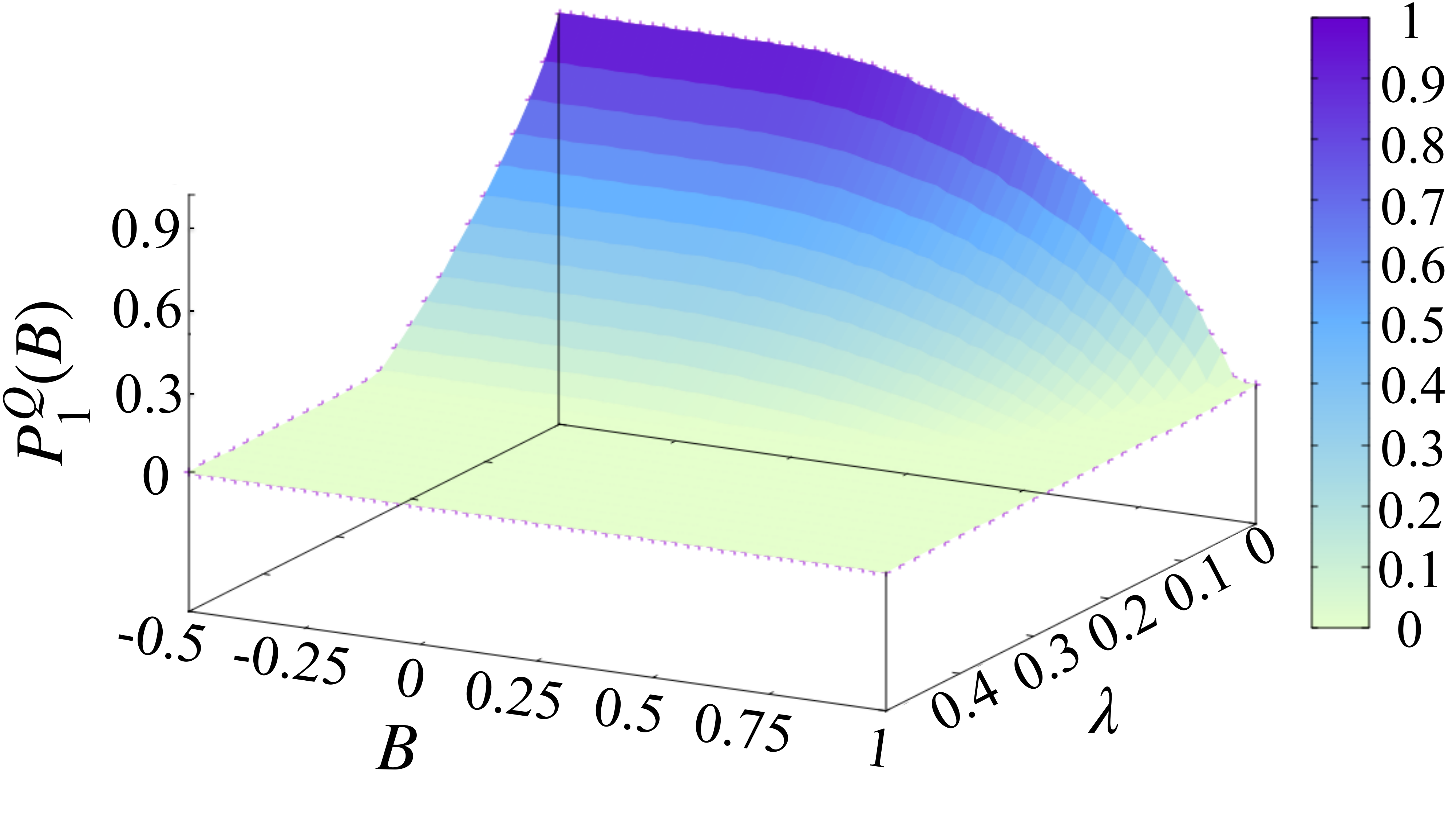}
    
    \caption{Variation of private channel capacity-power function $P_{1}^{Q}(B)$ for depolarizing channel with the channel noise parameter $\lambda$.}
    \label{fig:P(B)_for _depol_3d}
\end{figure}

\begin{figure}
    \centering
    \includegraphics[scale = 0.13]{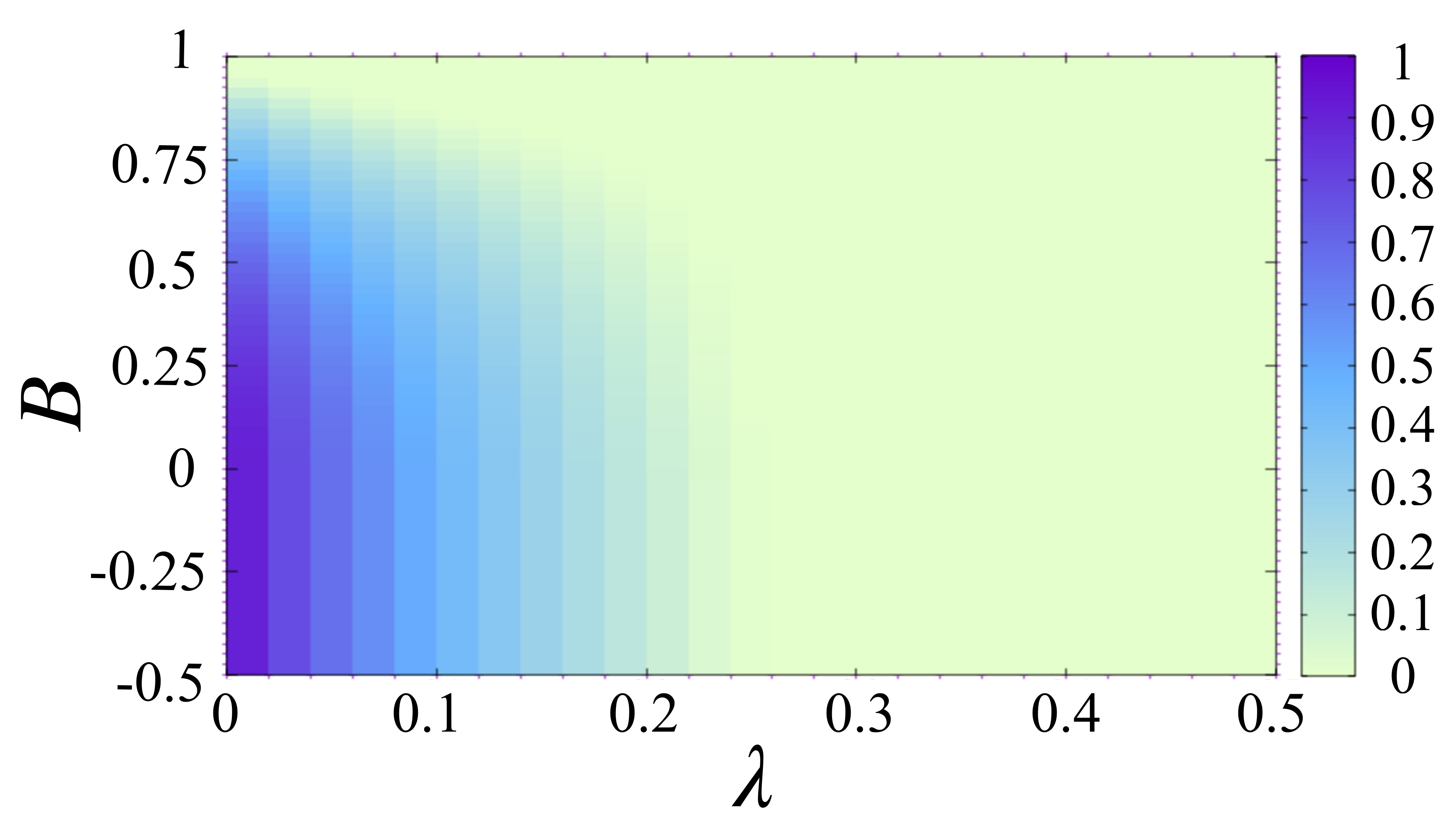}
    
    \caption{Projection of private channel capacity-power function $P_{1}^{Q}(B)$ for depolarizing channel on $B$-$\lambda$ plane.}
    \label{fig:P(B)_for_depol_map}
\end{figure}
The concavity and the additivity properties for the private capacity-power function will still hold if we have multiple compatible constraints similar to what we discussed for unassisted cases in Sec.~\ref{multiple_cons}.

\section{Discussion and Comments}
Quantum information is the information encoded in quantum states. Quantum information, quantified suitably, captures the communication capacity one can achieve with quantum states between a sender and a receiver. In this work, we have considered the scenario where we desire simultaneous transmission of energy and information. 

We establish the concavity and additivity properties of capacity for classical-quantum channels and for private communication through any degradable classical-quantum channel with a minimum output energy threshold. These properties are also true in the case of a noiseless channel with  pure states as input. 
For more general channels, we show that the capacity-power function is piecewise concave. 
Additionally, we perform numerical analysis of the unassisted  and private capacity, providing quantitative results that enrich our understanding of these channels' information-carrying capabilities. These theoretical and numerical findings offer key insights into quantum channels under power delivery constraints. We give an analytical expression of the capacity-power function for a noiseless channel. We use properties of random quantum states and the notion of typicality to argue that the capacity-power function can be obtained from the typical distribution which obeys an energy constraint. 
 
 This foundation can be built upon to study the impact of energy constraints in other aspects of quantum communication, such as entanglement generation and quantum error correction. Understanding how energy requirements affect these crucial tasks can contribute to the development of resource-aware quantum information protocols.
Further, this gives a natural extension to unidirectional bipartite communication over quantum channels under additional physical constraints.  
Advancing this program will explore optimized communication in multipartite scenarios, especially with multiple physical constraints and degrees of freedom. For example, the constraint can involve not only a bound on energy but also on the expectation value of multiple Hermitian operators as well as multiple players.
Investigating such tradeoffs among energy, classical capacity, and private capacity can provide valuable insights into fundamental limits and optimization strategies in quantum information processing.

 It is also of interest to extend Blahut-Arimoto techniques \cite{Blahut,arimoto,effifcient_approx_channel_capa_classical,effifcient_approx_channel_capa_quantum,bla_ari_type,quant_blahut_ari,eff_bound_unknown_channel} for quantum capacity-power functions to efficiently optimize over states. The introduction of random matrix theory into this channel capacity framework can give us insights  and help us to build more understanding of channel capacity without going to rigorous optimization techniques.

\section*{Acknowledgment}
This work was supported in part by grant  DST/ICPS/QusT/Theme-3/2019/Q69, New faculty Seed Grant from IIT Madras, and National Science Foundation grant PHY-2112890. The authors were further supported, in part, by a grant from Mphasis to the Centre for Quantum Information, Communication, and Computing (CQuICC) at IIT Madras.

\nocite{*}

\bibliographystyle{unsrt}
\bibliography{quantum_info_energy}

\end{document}